\documentclass[11pt,oneside,english,reqno]{amsart}
\usepackage[T1]{fontenc}
\usepackage[latin9]{inputenc}
\usepackage[a4paper]{geometry}
\usepackage{mathrsfs}
\usepackage{amstext}
\usepackage{amsthm}
\usepackage{amssymb}
\usepackage{mathtools}
\usepackage{xcolor}
\usepackage{hyperref}
\usepackage{bm}

\usepackage{dsfont}
\usepackage{enumerate}
\usepackage{verbatim} 
\usepackage{enumitem}
\usepackage{bbm}
\usepackage{todonotes}
\usepackage{subcaption}
\usepackage{float}
\usepackage{caption}

\makeatletter
\numberwithin{equation}{section}
\numberwithin{figure}{section}
\theoremstyle{plain}
\newtheorem{thm}{\protect\theoremname}[section]
\theoremstyle{definition}

\theoremstyle{remark}
\newtheorem{rem}[thm]{\protect\remarkname}
\theoremstyle{plain}
\newtheorem{lem}[thm]{\protect\lemmaname}
\theoremstyle{plain}

\theoremstyle{plain}
\newtheorem{cor}[thm]{\protect\corollaryname}
\theoremstyle{plain}
\newtheorem{hyp}[thm]{\protect\hypothesisname}
\theoremstyle{plain}
\newtheorem{ex}[thm]{\protect\examplename}

\makeatother

\usepackage{babel}
\providecommand{\corollaryname}{Corollary}
\providecommand{\definitionname}{Definition}
\providecommand{\lemmaname}{Lemma}
\providecommand{\propositionname}{Proposition}
\providecommand{\remarkname}{Remark}
\providecommand{\theoremname}{Theorem}
\providecommand{\hypothesisname}{Hypothesis}
\providecommand{\examplename}{Example}

\newcommand{\hatW}{\widehat{W}}

\newcommand{\cD}{\mathcal{D}}
\newcommand{\cE}{\mathcal{E}}
\newcommand{\cF}{\mathcal{F}}

\newcommand{\cK}{\mathcal{K}}

\newcommand{\EE}{\mathbb{E}}

\newcommand{\NN}{\mathbb{N}}

\newcommand{\PP}{\mathbb{P}}
\newcommand{\QQ}{\mathbb{Q}}
\newcommand{\RR}{\mathbb{R}}

\newcommand{\dd}{\mathrm{d}}
\newcommand{\coloneqq}{:=}
\newcommand{\var}{\mathrm{Var}}
\newcommand{\Erm}{\mathrm{E}}
\newcommand{\f}{\frac}

\DeclareMathOperator{\cov}{cov}

\begin{document}
 \setcounter{tocdepth}{1}

 \hypersetup{
     colorlinks   = true,
     citecolor    = blue,
     linkcolor    = blue
}

\title{Log-modulated rough stochastic volatility models}
\author{C. Bayer \and F. Harang \and P. Pigato }
\date{\today}


\thanks{\emph{Acknowledgments.} We are grateful to M. Fukasawa, J. Gatheral and A. Gulisashvili for valuable comments and support.
F. Harang gratefully acknowledges financial support from the STORM project 274410, funded by the Research Council of Norway. C.~Bayer gratefully acknowledges support by the German research council DFG via the cluster of excellence MATH+, project AA4-2.}

\address{Christian Bayer: Weierstrass Institute of Applied Analysis and Stochastics, Mohrenstr. 39, 10117 Berlin, Germany}
\email{christian.bayer@wias-berlin.de}

\address{Fabian A. Harang:
 Department of Mathematics, University of Oslo, P.O. box 1053, Blindern, 0316, OSLO, Norway}
\email{fabianah@math.uio.no} 

\address{Paolo Pigato: 
Department of Economics and Finance, University of Rome Tor Vergata, Via Columbia 2, 00133 Roma, Italy}
\email{paolo.pigato@uniroma2.it}

\begin{abstract}
    We propose a new class of rough stochastic volatility models obtained by modulating the power-law kernel defining the fractional Brownian motion (fBm) by a logarithmic term, such that the kernel retains square integrability even in the limit case of vanishing Hurst index $H$. The so-obtained log-modulated fractional Brownian motion (log-fBm) is a continuous Gaussian process even for $H = 0$. As a consequence, the resulting \emph{super-rough stochastic volatility models} can be analysed over the whole range $0 \le H < 1/2$ without the need of further normalization. We obtain skew asymptotics of the form $\log(1/T)^{-p} T^{H-1/2}$ as $T\to 0$, $H \ge 0$, so no flattening of the skew occurs as $H \to 0$.
\end{abstract}

\maketitle
%\tableofcontents

%{
%\hypersetup{linkcolor=black}
% \tableofcontents 
%}

\section{Introduction}

Prompted by new insights about the regularity of instantaneous variance obtained from realized variance data (see \cite{GJR18,BLP16,FTW19}), \emph{rough} stochastic volatility models have become more and more popular in the financial literature. Loosely speaking, these are stochastic volatility models
\begin{equation}
    \label{eq:stoch-vol}
    \dd S_t = S_t \sqrt{v_t} \dd B_t,
\end{equation}
where the logarithm of the instantaneous variance process $v$ roughly behaves like a fractional Brownian motion (fBm) with Hurst index $0 < H < 1/2$. One of the attractive features of rough volatility models is that they can explain the long-established power-law explosion of the ATM skew of options as time-to-maturity $T \to 0$ and, thus, provide excellent fits to the implied volatility surface, as was observed in \cite{BFG16}, but already anticipated much earlier in \cite{ALV07,fukasawa2011}. Hence, rough volatility models provide a framework which allows to get excellent fits to market data simultaneously w.r.t. to time series of prices of the underlying and to option prices, with few parameters.

Popular rough volatility models are either explicitly defined in terms of fBm, or rather in terms of a Volterra equation. Examples of the former case include the \emph{rough Bergomi} model of \cite{BFG16}, where the variance process is given  of the form
\begin{equation}
    \label{eq:rbergomi}
    v_t = \xi(t) \exp\left( \eta W^H_t - \f{\eta^2}{2} t^{2H}  \right).
\end{equation}
Here $\xi(t)$ denotes the forward variance curve, $W^H_t$ denotes the \emph{Riemann-Liouville} fBm, i.e., the Volterra process defined by
\begin{equation}
    \label{eq:RL-kernel}
    W^H_t \coloneqq \int_0^t K(t-s) \dd W_s, \quad K(r) \coloneqq \sqrt{2H} r^{H-1/2}, \quad r>0,
\end{equation}
where $W$ denotes a standard Bm correlated with the Bm $B$ with correlation coefficient $\rho$. As an example for the second type of model, in \cite{EuchRosenbaum2019} the authors  consider a \emph{rough Heston} model, where
\begin{equation}\label{eq:rCIR}
    v_t = v_0 + \f{1}{\Gamma(H+1/2)} \int_0^t (t-s)^{H-1/2} \lambda (\theta - v_s) \dd s + \f{1}{\Gamma(H+1/2)} \int_0^t (t-s)^{H-1/2} \lambda \nu \sqrt{v_s} \dd W_s.
\end{equation}
We note that the roughness of the fBm (or the singularity of the Volterra kernel  in \eqref{eq:RL-kernel} and \eqref{eq:rCIR}) causes considerable analytical and numerical difficulties, owing to the fact that the variance process $v$ fails to be a semimartingale or a Markov process in rough volatility models. Due to these technical difficulties, results holding for both the aforementioned classes of models are difficult to achieve. We refer to \cite{BFGMS19,FGP21,KRLP18} for attempts at unifying the treatment of rough volatility models.

Empirical studies of realized variance data as well as studies of the ATM skew in implied volatility surfaces tend to conclude that $H \ll 1/2$, often even $H < 0.1$. As both attempts involve a certain kind of smoothing -- realized variance being an estimate of $\int_t^{t+h} v_s \dd s$ rather than $v_t$ itself, option prices and their implied skews being in general not  available or reliable very close to maturity -- this begs the question, if $H$ actually might even be equal to zero. 
From the realized variance viewpoint, \cite{FTW19} indeed seems to suggest that $H$ could be $0$. Of course, $H=0$ is not allowed in the rough volatility models suggested above, but the case has been studied before in the literature on \emph{Gaussian multiplicative chaos}, see for instance the review paper \cite{RV14}. Indeed, a proper scaling limit of fBm $W^H$ as $H \to 0$ produces a log-correlated Gaussian field (see, for instance, \cite{NR18,hager2020multiplicative}).

Despite the well-established literature, some important financial questions regarding the $H \to 0$ limit are not very well understood yet. In particular, what happens with the ATM skew of implied volatility as $H \to 0$. On the one hand, given that the skew behaves like $T^{H-1/2}$ as time-to-maturity $T \to 0$ in rough volatility models with $H>0$, one might expect a power law explosion as $T^{-1/2}$ in the limiting case $H = 0$. However, a closer look at the asymptotic results for $H>0$,  casts some doubt on this conjecture. Indeed, taking the moderate deviation asymptotics of \cite{BFGHS19} as one example of such an expansion, we have the asymptotic formula
\begin{equation}
    \label{eq:moderate-skew}
    \text{skew} \sim \text{const}\, \rho \eta \f{\sqrt{2H}}{(H+1/2)(H+3/2)} T^{H-1/2}
\end{equation}
as $T \to 0$. Of course, the factor $\f{\sqrt{2H}}{(H+1/2)(H+3/2)} \to 0$ as $H \to 0$, so that \eqref{eq:moderate-skew} entails two limits ($H\to0$, $T\to0$), which cannot necessarily be interchanged. Note that $\sqrt{2H}$ appears in \eqref{eq:moderate-skew} by requiring the underlying fBm to have variance equal to one at time $t=1$. Indeed, some standardization of this type is needed in order to make models for different values of $H$ comparable -- even though the choice of standardization may be quite important.

\begin{rem}
As described above, in this paper we vary the Hurst index $H$ while keeping the other model parameters -- in particular, the vol-of-vol $\eta$ -- fixed. An alternative point of view motivated from the shape of the skew itself is to keep $\sqrt{H} \eta$ rather then $\eta$ fixed, which leads to more stable behavior of the skew. The second alternative, however, has undesirable effects on other properties of the model. Fixing $\sqrt{H}\eta$ implies exploding variance of $\log v_t$ in the model~\eqref{eq:rbergomi} as $H \to 0$. Consequently, we expect an explosion of the kurtosis of the asset price as well as of the volatility of VIX options. 
\end{rem}

 The multiplicative-chaos approach in  \cite{NR18} is used in \cite{forde_bergomi_H0} to establish a $H\to 0$ limit for rough Bergomi, for which the limit skewness vanishes or blows-up depending on the renormalization. Using continuity of Volterra integral equations, a $H\to 0$ limit for driftless rough Heston is considered in \cite{forde_heston_H0}, in this case with a non-symmetric limit behavior. However, in \cite{forde_bergomi_H0,forde_heston_H0} no explicit formula for the skew of implied volatility is given. Moreover, in both cases the limit volatility is not a process, but is defined as a distribution. Hyper-rough volatility, in a sense analogous to a $H<0$ model, has also been considered \cite{JusselinRosenbaum2020,jaber2019weak}, but also in this case spot volatility is not defined.  
 The extreme $T^{-1/2}$ speed of explosion for the skew expected in the $H\to 0$ limit has been shown to be a model-free bound \cite{lee05,fuk2010}, and is reached under local volatility through a volatility function with a singularity ATM \cite{pig19}, but this poses the problem of time-consistency (see also \cite{fps2020}). To the best of our knowledge, this extreme behavior of the skew has not been shown for any (time-consistent) stochastic volatility model, where the volatility is a proper process.

\subsection{Our contribution.} In this paper, we consider an actual process with $H=0$ by introducing a logarithmic term in the definition of the kernel $K:\RR_+\rightarrow \RR$, which for small $r>0$ behaves similarly to $r^{H-\frac{1}{2}}\log(1/r)^{-p}$ for some parameter $p>1$. This modification ensures   that $K$ remains square integrable for all $H \in [0,1/2)$, see \eqref{eq:log-kernel} for the precise definition. Note that we ignore $H\geq 1/2$ in this paper, as what we are interested in is the $H\to 0$ limit, but there would be no real difficulties in considering $H \in [0,1]$, or even $H > 1$. Hence, the resulting family of Gaussian Volterra processes $\hatW$ will be continuous and with finite variance even for $H=0$, and the ambiguities of the asymptotic analysis for $H\to0$ and $T\to0$ cease to matter, as we can simply do the asymptotic for $H=0$. We stress again that $\hatW$ is a proper, continuous Gaussian process even for $H=0$. At the same time, as we apply our logarithmic modification only close to the singularity of the power-law kernel, we may expect that the resulting rough volatility models are close to the corresponding standard rough volatility models for $H \gg 0$, see Figure~\ref{fig:skew_sberg_rberg}. The process we propose here can be seen as an extension of the log Brownian motion studied in \cite{MocioalcaViens2005}, to include a fractional power. This allows for a better comparison with classical fractional processes, such as the Riemann-Liouville fractional Brownian motion, typically used in rough volatility models. We also mention that the standard log-Brownian motion (without the fractional power) has recently been analysed in the context of rough volatility models in \cite{gulisashvili2020timeinhomogeneous}, as well as in the context of regularization by noise for ill-posed ODEs in \cite{harang2020cinfinity}.
\begin{figure}
    \centering
    \begin{subfigure}[t]{0.49\textwidth}
        \centering
        \includegraphics[width=\textwidth]{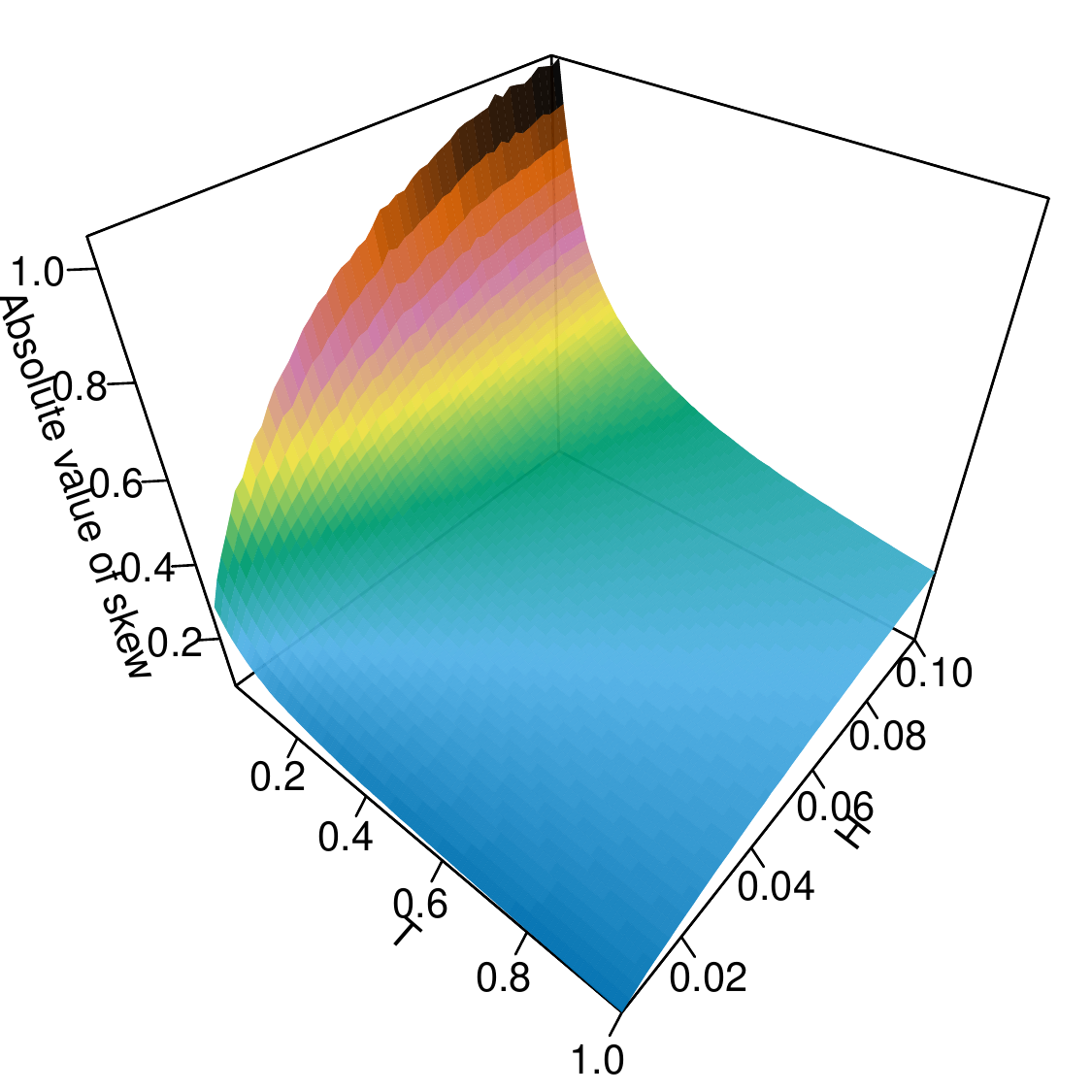}
        \caption{Rough Bergomi model}
        \label{fig:skew-surface-rBergomi}
    \end{subfigure}~
    \begin{subfigure}[t]{0.49\textwidth}
        \centering
        \includegraphics[width=\textwidth]{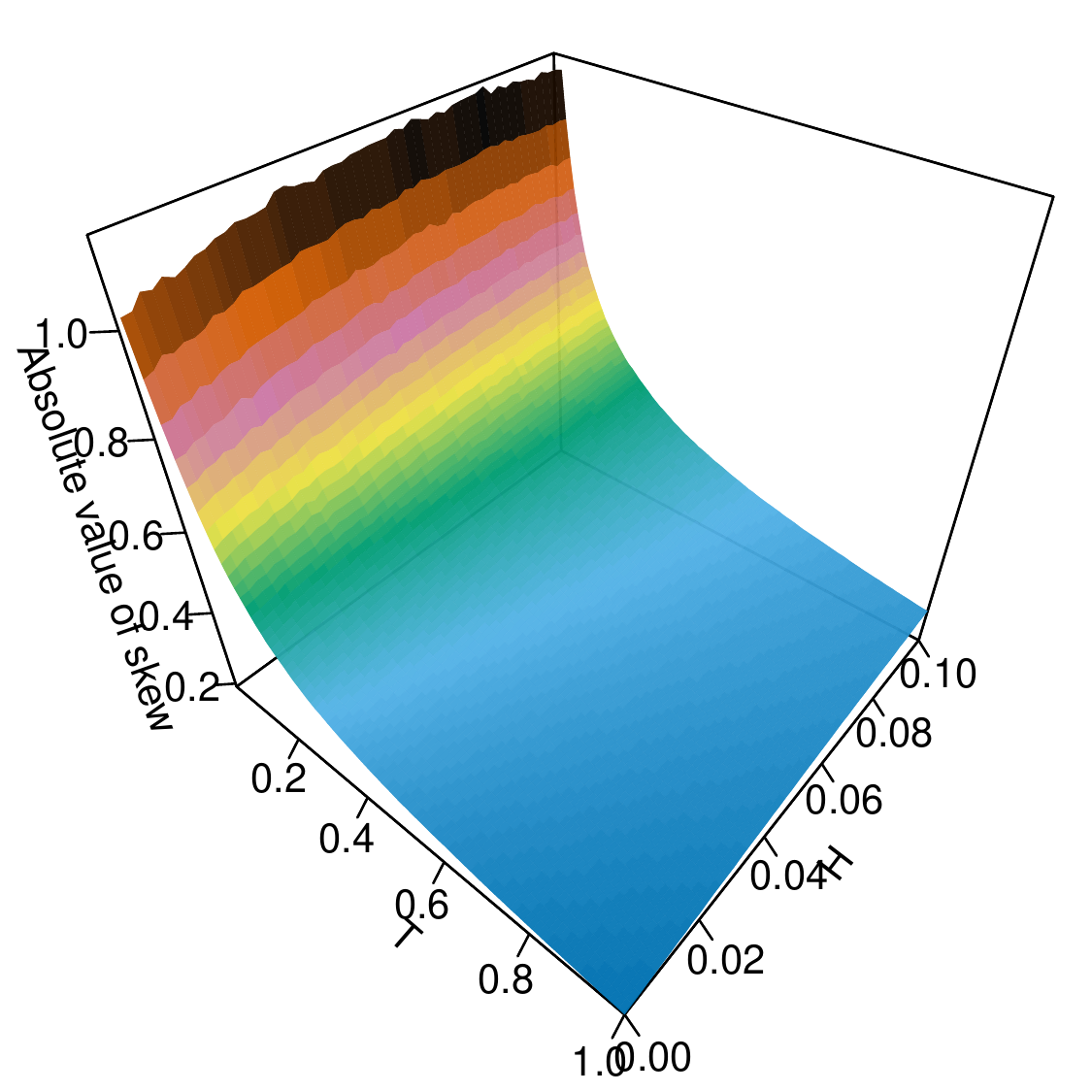}
        \caption{Super-rough Bergomi model}
        \label{fig:skew-surface-sBergomi}
    \end{subfigure}
    \caption{ATM implied volatility skews (absolute values) in the (super-) rough Bergomi model plotted against expiry $t$ and Hurst index $H$. Skews are computed by Monte Carlo simulation based on exact simulation of the underlying (log-modulated) fBm. Note that $H = 0$ is included in the plot in the super-rough case. Parameter values of the rough Bergomi model are $\eta = 2.0$, $\rho = -0.7$, $\xi(t) \equiv 0.04$. Additional parameters for the log-fBm (see Section~\ref{sec:log-fbm} for details) are $\zeta = 0.1$, $p = 2.0$. Note how this seems in keeping with the findings in \cite{forde_bergomi_H0}, of a vanishing skewness as $H\downarrow 0$ in rough Bergomi.
    }
    \label{fig:skew-surface}
\end{figure}

In this way, we are able to obtain rough volatility models which allow continuous interpolation for $H \in [0,1/2)$, in the sense that all such choices of $H$ are valid within the same model, with no apparent breaks between them. To illustrate this observation, we consider a \emph{super-rough Bergomi model}, which is simply obtained by replacing the Riemann-Liouville fBm by the log-fBm defined in \eqref{eq:log-fbm} below in the rough Bergomi model of \cite{BFG16}. Figure~\ref{fig:skew-surface-sBergomi} shows the ATM-skew for various expires and values of $H$ between -- and including -- $0$ and $0.1$. Indeed, the surface ``looks'' smooth in $H$, visually indicating a smooth transition from the power law explosion $T^{H-1/2}$ for $H>0$ to the skew behaviour at $H = 0$. In contrast, the skew-behaviour changes remarkably for the standard rough Bergomi model for small $H$, see Figure~\ref{fig:skew-surface-rBergomi}. In particular, the skew flattens significantly for very small $H$. On the other hand, the log-modulated version in Figure~\ref{fig:skew-surface-sBergomi} shows no signs of flattening. To the contrary, a more refined analysis, which is the main purpose of this paper, shows that the skew behaves like $T^{H-1/2}$ -- up to logarithmic terms --  and, hence, steepens as $H \to 0$.

Note that the log-fBm does not have a scale invariance property, thereby making any short time asymptotics very difficult. Hence, in this paper we first use the vol-of-vol expansion in \cite{fukasawa2011} to obtain an asymptotic formula for the ATM skew when the volatility-of-volatility is small. Indeed, we obtain a skew formula of the form
\begin{equation}
    \label{eq:skew-asymptotic}
    \text{skew} \approx a_{H,\zeta,p}\, \rho \,\log(1/T)^{-p} \, T^{H-1/2} \epsilon, \text{ as } T \to 0,
\end{equation}
for small vol-of-vol $\epsilon$, Hurst parameter $H\in[0,1/2)$, see Theorem~\ref{thr:skew}. Here,  $p>1$ is a parameter of the kernel defined in \eqref{eq:log-kernel}, and $a_{H,\zeta,p}$ is a constant depending on $H$ -- and other parameters -- which is smooth in $H$ with $a_{0,\zeta,p} \neq 0$. 
Then, we prove that the short-time asymptotics corresponding to \eqref{eq:skew-asymptotic} at the Edgeworth CLT regime holds even 
without considering the small vol-of-vol regime, for log-modulated models (in the sense of regular variation) with $H>0$.

\subsection{Outline of the paper}
In Section \ref{sec:log-fbm} we introduce a class of Gaussian processes, extending the notion of fractional Brownian motion through a modulation with a $\log$ term.  In Section \ref{sec:moments-log-fract} we compute some essential probabilistic features of the log fractional Brownian motion (log-fBm) such as variance and covariance, which will be crucial for applications to asymptotic expansions of the implied volatility corresponding to certain rough volatility models as well as for simulation of the processes. In Section \ref{sec:Fukasawas method} we provide a short overview of the martingale expansion developed by Fukasawa in \cite{fukasawa2011}, and its application towards analysis of the implied volatility surface. Furthermore, we provide explicit computations of the covariance terms appearing in the asymptotic expansion in the case when the volatility is driven by a log-fBm. Section~\ref{sec:skew-expansion} deals with a particular skew expansion, asymptotic in vol-of-vol,  using Fukasawa's martingale approach. Here, we also include an asymptotic expansion for the rough Bergomi model, when driven by a log-fBm. In Section \ref{sec: asymptotic paolo} we consider a slightly more general kernel and the asymptotics for the skew at the Edgeworth CLT regime, that holds for any vol-of-vol parameter, generalising a result in \cite{fukasawa2020}. At last, in Section \ref{sec:numerics}  we provide some details on numerical simulations and  computations of the skew. 

\section{Rough and super rough  volatility modelling}
\label{sec:log-fbm}

The fractional Brownian motion (fBm) is a well studied Gaussian process. A simplified version of this process, called the Riemann-Liouville fBm is given as a Volterra type stochastic integral with respect to a Brownian motion, i.e. 
\begin{equation*}
    W^H_t 
    := \sqrt{2H}\int_0^t(t-s)^{H-\frac{1}{2}}\dd W_s, 
\end{equation*}
where $\{W_t\}_{t\in 0,T]}$ denotes a standard Brownian motion. However this process is typically defined for $H\in (0,1)$, and thus excludes the case when $H=0$.  To overcome this challenge, we propose to modulate the Riemann-Liouville fBm with a log term to  control the singularity in the kernel $(t-s)^{H-\frac{1}{2}}$.   
In this section we therefore will construct a particular fractional process which allows to generalize the Riemann-Liouville fBm to $H\in [0,1)$. 
We consider the Gaussian  Volterra  process
\begin{equation}
  \label{eq:log-fbm}
  \widehat{W}_t := \int_0^t K(t-s) \dd W_s,
\end{equation}
where the kernel $K$ satisfies
\begin{equation}
  \label{eq:log-kernel}
  K(r) \coloneqq C r^{H-1/2} \max\left(\zeta \log\left( 1 / r \right), \,
    1 \right)^{-p} =
  \begin{cases}
    C r^{H-1/2} \zeta^{-p} \log(1/r)^{-p}, & 0 \le r \le 
    e^{-1/\zeta},\\
    C r^{H-1/2}, & r > e^{-1/\zeta}.
  \end{cases}
\end{equation}
We assume that $0 \le H < 1/2$, $\zeta > 0$, and $p>1$. $C$ is a
constant which will be chosen to normalize the process, i.e., to guarantee
\begin{equation*}
  \var(\widehat{W}_1) = 1.
\end{equation*}
\begin{figure}
    \centering
    \includegraphics{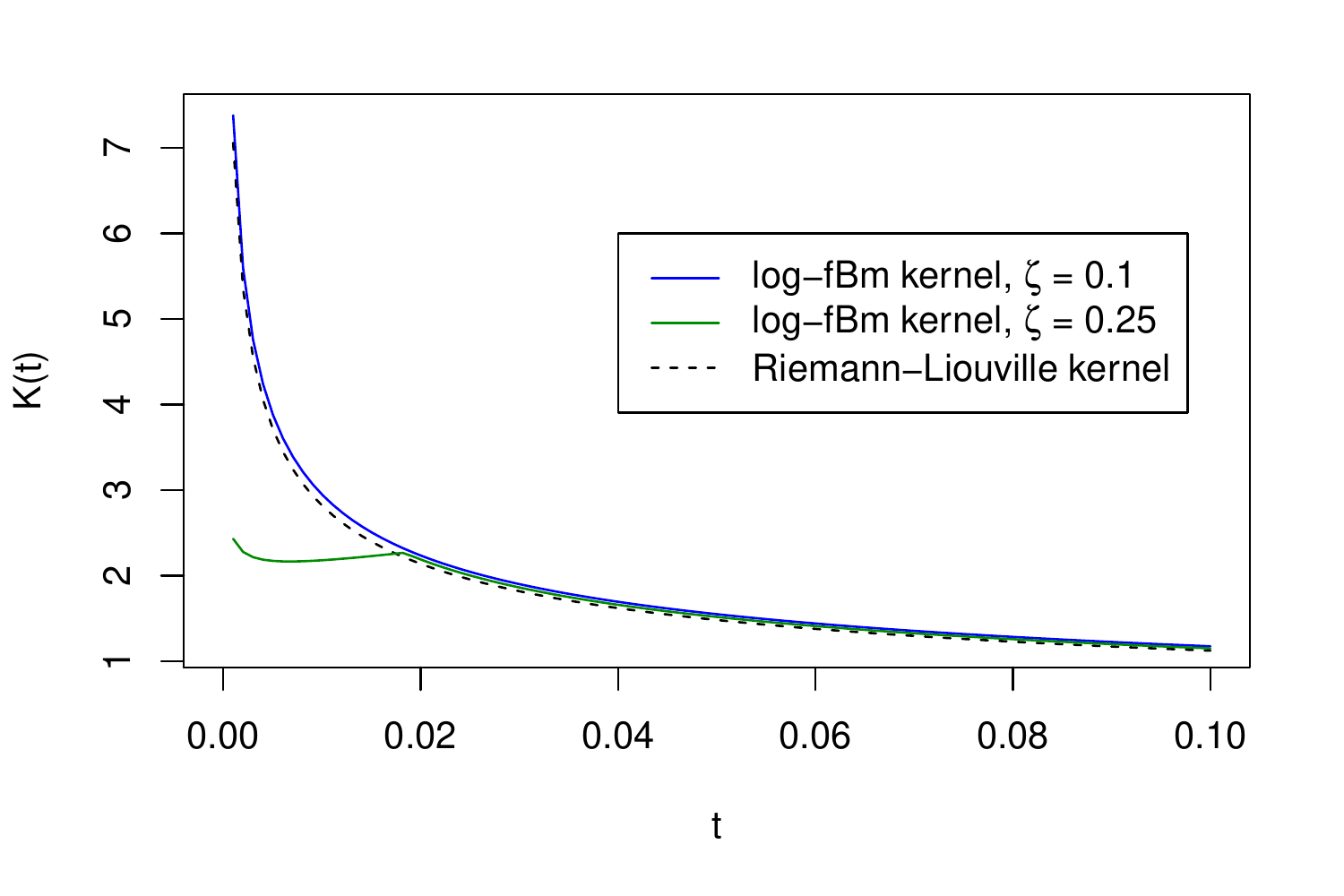}
    \caption{The logarithmic kernel \eqref{eq:log-kernel} with $H = 0.1$,  $p = 2$, and $\zeta \in \{0.1, 0.25\}$ compared against the Riemann-Liouville kernel $K(t) \simeq t^{H-1/2}$. All kernels are normalized to have $\var \, \widehat{W}_1 = 1$.}
    \label{fig:my_label}
\end{figure}
\noindent
Therefore, $C$ will depend on all the other parameters. We call the process  $\{\widehat{W}_t\}_{t\in 0,T]}$ a {\em log-fractional Brownian motion (log-fBm)}.  We note that by the
choice of these parameters, $\{\hatW_t\}_{t\in 0,T]}$ is a continuous Gaussian process with
vanishing expectation. Indeed, it is readily checked that $K\in L^2([0,T])$ (see in particular Lemma \ref{lem:var-logfbm} for explicit computations), and thus $\{\widehat{W}_t\}_{t\in 0,T]}$ is well defined as a Wiener integral. Moreover, due to the assumption that $p>1$,
the continuity can be verified by Fernique's continuity condition (see \cite{MocioalcaViens2005} below Definition 18 or  \cite{gulisashvili2020timeinhomogeneous} Remark 3.4), even in the case $H=0$. In the case $H>0$, sample paths can also be proven to be of the same regularity as the fractional Brownian motion, in the sense of H\"older continuity. Indeed, in Lemma \ref{lem:Incremental varaince} we prove that there exists a constant $C=C_{\zeta,H}>0$ such that 
\begin{equation}
    \EE[|\widehat{W}_t-\widehat{W}_s|^2]\leq C |t-s|^{2H},
\end{equation}
and thus an application of Kolmogorov's continuity theorem yields the claimed regularity. 

For ease of notation we introduce $\chi \coloneqq  e^{-1/\zeta}$. 

\begin{rem}
As $\chi$ depends exponentially on $1/\zeta$ and the log-fBm-kernel introduced in \eqref{eq:log-kernel} only differs from the standard Riemann-Liouville kernel on $(0,\chi)$, one may be tempted to expect that the corresponding process $\hatW$ behaves very similarly to the Riemann-Liouville fBm as often used in rough volatility models. This is undoubtedly true for $H \gg 0$, and motivates the whole paper, but note that there are profound differences as $H \to 0$, as witnessed by Figure~\ref{fig:scaling_factor}. In particular, despite the very localized changes, the log-fBm has finite variance even for $H  = 0$.
\end{rem}

We see the super-rough Bergomi model mostly as a \emph{perturbation} of the rough Bergomi model, which stays close to the rough Bergomi model when $H \gg 0$, but still has nice properties for $H \to 0$, see Figure~\ref{fig:skew_sberg_rberg} for a comparison of skews in the (super-) rough Bergomi model. This implies that the super-rough Bergomi model differs substantially from the rough Bergomi model as $H \to 0$, as already seen in Figure~\ref{fig:skew-surface}.

\begin{rem}
\label{rem:parameter-choice}
The super-rough Bergomi model adds two more parameters ($\zeta>0$ and $p>1$) to the rough Bergomi model. If we want to keep close to the rough Bergomi model for not-too-small $H$, then we need both $\zeta$ and $p$ to be chosen small within their admissible ranges.  This, however, may very well introduce numerical difficulties for $H \approx 0$, as the kernel approaches a kernel which fails to be square integrable as $\zeta \to 0$ or $p \to 1$. For financial practise, we suggest fixing $\zeta$ to a convenient value, e.g., $\zeta = 1/10$, and calibrating $p$ using the small-time skew asymptotic of Corollary~\ref{cor:rbergomi-short-skew}. Alternatively, one could additionally fix $p$, e.g., to $p = 2$, in which case one should view the log-modulation as a regularization technique without inherent financial meaning. In either case, log-modulation is probably only sensible if $H$ is very small.
\end{rem}

\section{Moments of the log-fractional Brownian motion}
\label{sec:moments-log-fract}

The skew formulas to be derived in later section will depend on formulas for some moments of the log-fractional Brownian motion and the underlying Brownian motion. Computing these moments will also give us an explicit formula for the constant $C$ in the kernel \eqref{eq:log-kernel}. Throughout this section, we shall often use the following elementary lemma.

\begin{lem}
  \label{lem:exp-integral}
  Consider $0 < u < 1$, $a \le 1$, and $b>1$. Then we have
  \begin{equation*}
    \int_0^u r^{-a} \log\left( \frac{1}{r} \right)^{-b} \dd r = 
    \log\left(\frac{1}{u} \right)^{1-b} \Erm_b\left( (1-a) \log\left(\frac{1}{u} \right) \right),
  \end{equation*}
  where $E_b$ denotes the exponential integral, given by
  \begin{equation*}
    \Erm_b(x) \coloneqq \int_1^\infty e^{-xt} t^{-b} \dd t, \quad x \ge 0.
  \end{equation*}
\end{lem}

Note that the exponential integral $\Erm_b(x)$ is infinite for negative $x$, which is excluded by our
assumptions, and that $\Erm_b(x)$ can be expressed using the incomplete gamma function. Using the above lemma, second moments of $\{\hatW_t\}_{t\in 0,T]}$ can be computed explicitly.
\begin{lem}
  \label{lem:var-logfbm}
 Let $\{\hatW_t\}_{t\in 0,T]}$ be the log-fBm given in \eqref{eq:log-fbm} with kernel given in \eqref{eq:log-kernel}. Then  the variance of $\hatW$ satisfies
   \begin{equation*}
    \var (\hatW_t) =
    \begin{cases}
      C^2 \left[ \zeta^{-2p} \log\left(\f{1}{t \wedge \chi}
        \right)^{1-2p} \Erm_{2p}\left( 2H \log\left( \f{1}{t \wedge \chi}
          \right) \right) + \f{t^{2H} - (t \wedge \chi)^{2H}}{2H} \right],
      & H > 0,\\
      C^2\left[ \zeta^{-2p}  \f{1}{2p-1} \log\left(\f{1}{t \wedge \chi}
        \right)^{1-2p} + \log\left( \f{t}{t \wedge \chi} \right) \right],
      & H = 0.
    \end{cases}
  \end{equation*}
  Assuming $\chi < 1$, the scaling constant $C = C_{H,\zeta,p}$ required to ensure $\var\, \hatW_1 = 1$ satisfies
  \begin{equation*}
      C_{H,\zeta,p} \coloneqq \begin{cases}
      \left[ \f{1}{\zeta} \Erm_{2p}\left( 2H/\zeta \right) + \f{1 - \chi^{2H}}{2H}\right]^{-1/2}, & H > 0,\\
      \left[  \f{2p}{(2p-1)\zeta} \right]^{-1/2}, & H = 0.
      \end{cases}
  \end{equation*}
\end{lem}
\begin{proof}
  By definition, we have 
  \begin{align*}
  \var (\hatW_t) &= \int_0^t K(t-r)^2 \dd r\\
               &= \int_0^t K(r)^2 \dd r\\
               &= C^2 \left[ \zeta^{-2p} \int_0^{t \wedge \chi} r^{2H-1} \log\left( \f{1}{r}
                 \right)^{-2p} \dd r +
                 \int_{t \wedge \chi}^t r^{2H-1} \dd r\right].
  \end{align*}
 Applying  Lemma~\ref{lem:exp-integral}, the integral gives
  \begin{equation*}
    \int_0^{t \wedge \chi} r^{2H-1} \log\left( \f{1}{r}
    \right)^{-2p} \dd r =  \log\left(\f{1}{t \wedge \chi}
    \right)^{1-2p} \Erm_{2p}\left( 2H \log\left( \f{1}{t \wedge \chi}
      \right) \right),
  \end{equation*}
  which simplifies in the case $H = 0$ to the expression
  \begin{equation*}
    \int_0^{t \wedge \chi} r^{-1} \log\left( \f{1}{r}
    \right)^{-2p} \dd r = \f{1}{2p-1} \log\left(\f{1}{t \wedge \chi}
    \right)^{1-2p}.
  \end{equation*}
  For the second integral, we have by standard computations
  \begin{equation*}
    \int_{t \wedge \chi}^t r^{2H-1} \dd r =
    \begin{cases}
      \f{t^{2H} - (t \wedge \chi)^{2H}}{2H}, & H > 0,\\
      \log(t) - \log\left( t \wedge \chi \right), & H = 0.
    \end{cases}\qedhere
  \end{equation*}
\end{proof}

\begin{figure}
    \centering
    \includegraphics{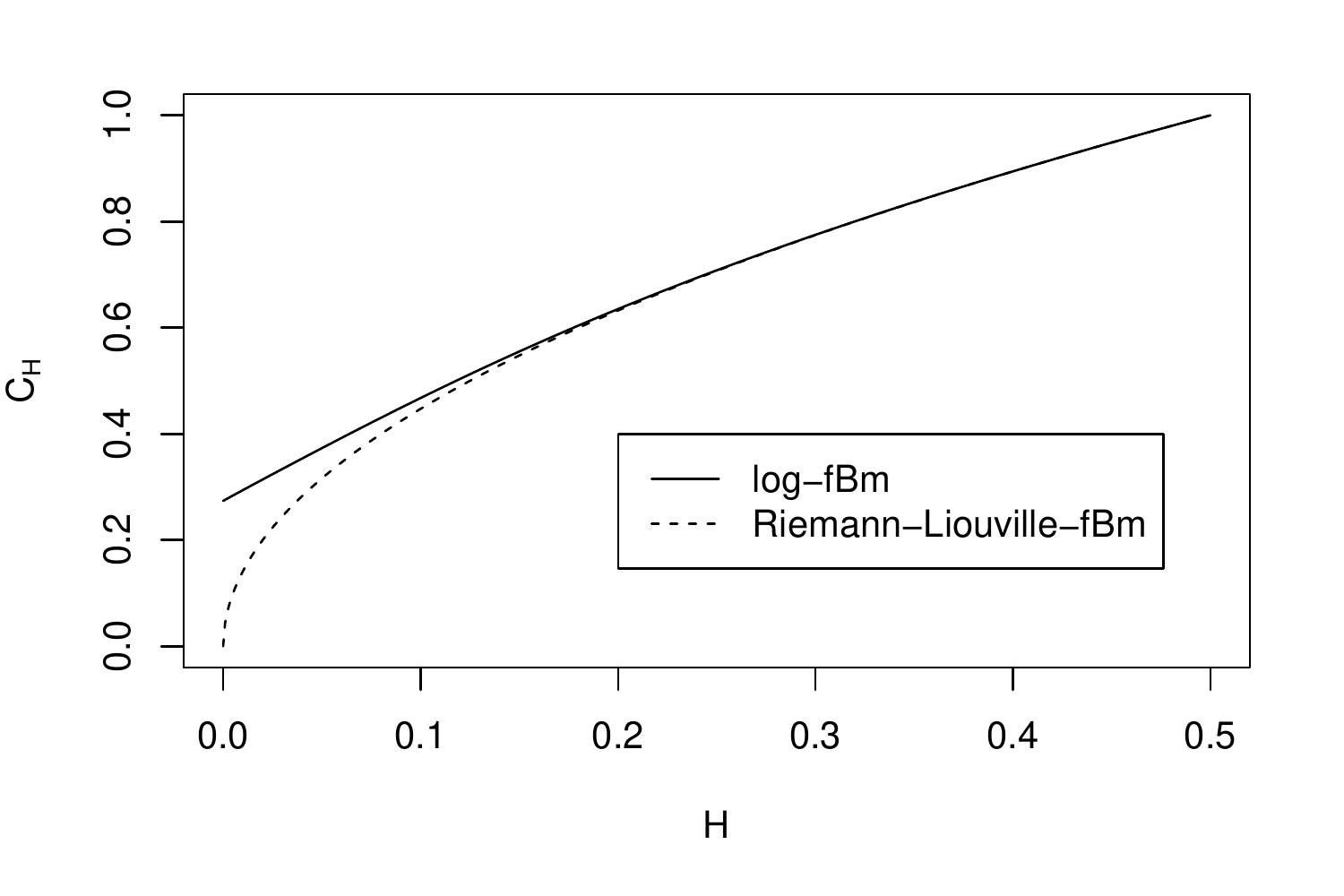}
    \caption{The scaling factor $C_{H,\zeta,p}$ needed to achieve $\var\,\hatW_1 = 1$ as shown in Lemma~\ref{lem:var-logfbm} compared to the corresponding scaling factor $\sqrt{2H}$ for the Riemann-Liouville kernel $K(r) \simeq r^{H-1/2}$. Parameters for the log-fBm are  $\zeta = 0.1$, $p = 2$, giving $\chi \approx 5 \times 10^{-5}$.}
    \label{fig:scaling_factor}
\end{figure}
\noindent
Note that the scaling factors are continuous in $H$ on $[0,1/2]$, see also Figure~\ref{fig:scaling_factor}.

Unfortunately, we have not been able to find closed form expressions for the covariances $\cov\left( \hatW_t, \, \hatW_s \right)$ of the log-fractional Bm. Nonetheless, numerical integration is relatively easy using the \emph{double exponential method} (see, for instance, \cite{MS01}) to take care of the singularity at the boundary of the integral.

On the other hand, we can give suitable bounds for the incremental variance in the case $H>0$, which ensures H\"older continuity of the sample paths by application of Kolmogorov's continuity theorem. This follows from the following lemma:
\begin{lem}\label{lem:Incremental varaince}
  Assume that $\zeta \le \frac{p}{1/2-H}$ and $H > 0$.
  For any $s,t\in [0,1]$ we have 
  \begin{equation}\label{eq:inc var}
      \EE[|\widehat{W}_t-\widehat{W}_s|^2]\leq C|t-s|^{2H}.
  \end{equation}
\end{lem}
\begin{proof}
  We can assume w.l.o.g.~that $ s<t$. Then
  \begin{equation}\label{eq:incremental-varince-kernel}
    \EE[|\widehat{W}_t-\widehat{W}_s|^2]=\int_s^t K(t,r)^2\dd r+\int_0^s( K(t,r)-K(s,r))^2\dd r \eqqcolon I^1(s,t) +I^2(s,t). 
  \end{equation}
  In the following, all constants $C, C_1, \ldots$ are assumed positive and may depend on $H, \zeta, p$, but not on $s$ or $t$. They may also change from line to line.
  We begin by considering $I^1(s,t)$. By definition of the kernel, we have
  \begin{equation*}
    I^1(s,t) = 
    \begin{cases}
      C_1 \int_s^t \log\left( \frac{1}{t-r} \right)^{-2p} (t-r)^{2H-1} \dd r, & t-s \le \chi,\\
      C_1 \int_{t-\chi}^t \log\left( \frac{1}{t-r} \right)^{-2p} (t-r)^{2H-1} \dd r + C_2 \int_s^{t-\chi} (t-r)^{2H-1} \dd r, & \text{else}.
    \end{cases}
  \end{equation*}
  In the case $t-s \le \chi$, we set $r=s+\theta(t-s)$, implying 
  \begin{align*}
    \left|\int_s^t (t-r)^{2H-1}\log\left(\frac{1}{t-r}\right)^{-2p}\dd r\right| &= |t-s|^{2H}\int_0^1 (1-\theta)^{2H-1}\log\bigg(\frac{1}{(t-s)(1-\theta)}\bigg)^{-2p}\dd \theta 
    \\ &\leq  |t-s|^{2H}\int_0^1 (1-\theta)^{2H-1}\log\bigg(\frac{1}{\chi(1-\theta)}\bigg)^{-2p}\dd \theta
  \end{align*}
  where we have used that $x\mapsto \log\big(\frac{1}{x(1-\theta)}\big)^{-2p}$ is increasing for each $\theta\in (0,1)$. The last integral is easily checked to be bounded. A slight modification of the argument gives the bound
  \begin{equation*}
    \int_{t-\chi}^t \log\left( \frac{1}{t-r} \right)^{-2p} (t-r)^{2H-1} \dd r \le C |t-s|^{2H}
  \end{equation*}
  in the second case, as well. Using the trivial estimate
  \begin{equation*}
    \int_s^{t-\chi} (t-r)^{2H-1} \dd r \le C |t-s|^{2H},
  \end{equation*}
  we obtain
  \begin{equation}\label{eq:inc-var-I1}
    |I^1(s,t)|\leq C |t-s|^{2H}.
  \end{equation}

  Regarding $I^2$, we need to bound three different terms. Indeed, the integration domain $[0,s]$ may -- depending on the parameters -- naturally split up into up to three subintervals:
  \begin{enumerate}
  \item For $t-\chi \le r \le s$, both kernels are log-modulated.
  \item For $s-\chi \le r < t-\chi$, $K(t,r)$ is a pure power-law kernel, but $K(s,r)$ is still log-modulated.
  \item Finally, for $0 \le r < s-\chi$, both kernels have the power-law form.
  \end{enumerate}
  More precisely, defining 
  \begin{align*}
    a &\coloneqq \min(\max(t-\chi, 0), s),\\
    b &\coloneqq \min(\max(s-\chi, 0), a),\\
  \end{align*}
  we have
  \begin{align*}
    I^2(s,t) &= \int_0^b \left[ K(t,r) - K(s,r) \right]^2 \dd r + \int_b^a \left[ K(t,r) - K(s,r) \right]^2 \dd r + \int_a^s \left[ K(t,r) - K(s,r) \right]^2 \dd r\\
             &= C \int_0^b \left[ (t-r)^{H-1/2} - (s-r)^{H-1/2} \right]^2 +\\
             &\quad + C \int_b^a \left[  (t-r)^{H-1/2} - \zeta^{-p} \log\left( \f{1}{s-r} \right)^{-p} (s-r)^{H-1/2} \right]^2 \dd r+\\
             &\quad+ C \zeta^{-2p} \int_a^s \left[\log\left( \f{1}{t-r} \right)^{-p} (t-r)^{H-1/2} - \log\left( \f{1}{s-r} \right)^{-p} (s-r)^{H-1/2} \right]^2 \dd r\\
    &\eqqcolon I^2_1(s,t) + I^2_2(s,t) + I^2_3(s,t).
  \end{align*}
  We shall now prove that each of the terms $I^2_1$, $I^2_2$, $I^2_3$ can be bounded by $C (t-s)^{2H}$.
  
  By the change of variables $y=\f{t-r}{t-s}$, we have that 
  \begin{equation*}
    I^2_1(s,t) = C (t-s)^{2H}\int_{(t-b)/(t-s)}^{t/(t-s)} \bigg( y^{H-\frac{1}{2}}-(y-1)^{H-\frac{1}{2}}\bigg)^2 \dd y.
  \end{equation*}
  Recall the following inequalities:
    $0\leq y^{H-\frac{1}{2}}-(y+1)^{H-\frac{1}{2}}\leq y^{H-\frac{3}{2}}$, for $y>1$, for $H\in(0,1/2]$
  and
    $0\leq 
    (y+1)^{H-\frac{1}{2}}
    -
    y^{H-\frac{1}{2}}\leq (y+1)^{H-\frac{3}{2}}$, for $y>1$, for $H\in[1/2,1)$. By the comparison test, these two inequalities imply in particular that for
  $H\in(0,1)$  
  \begin{equation*}
    \int_1^{\infty} ( y^{H-\frac{1}{2}}-(y-1)^{H-\frac{1}{2}})^2\dd y<\infty.
  \end{equation*}
Using also $\f{t-b}{t-s} \ge 1$, we obtain
  \begin{equation}
    \label{eq:I^2_1}
    I^2_1(s,t) \le C (t-s)^{2H}.
  \end{equation}

  We proceed to bounding $I^2_2(s,t)$. Note that by assumption, the kernel $s \mapsto K(s,r)$ is decreasing on $(r, \infty)$. Hence, $K(t,r) \le K(s,r)$. Moreover, the log-modulation factor $\zeta^{-p} \log\left( \f{1}{s-r} \right)^{-p} \le 1$, such that
  \begin{equation*}
    \left[K(t,r) - K(s,r) \right]^2 \le C \left[ (t-r)^{H-1/2} - (s-r)^{s-1/2} \right]^2.
  \end{equation*}
  Using the same calculation as for the estimate of $I^2_1$, we obtain
  \begin{equation}
    \label{eq:I^2_2}
    I^2_2(s,t) \le C(t-s)^{2H}.
  \end{equation}

  Finally, using
  \begin{equation*}
    \log\left( \f{1}{s-r} \right)^{-p} \le \log\left( \f{1}{t-r} \right)^{-p} \le \log\left( \f{1}{t-a} \right)^{-p} \le \zeta^{p}
  \end{equation*}
  for $a \le r \le s < t$, as well as $K(t,r) \le K(s,r)$, we can bound
  \begin{multline*}
    \left[\log\left( \f{1}{t-r} \right)^{-p} (t-r)^{H-1/2} - \log\left( \f{1}{s-r} \right)^{-p} (s-r)^{H-1/2} \right]^2 \\
    \le \zeta^{2p} \left[ (t-r)^{H-1/2} -  (s-r)^{H-1/2} \right]^2.
  \end{multline*}
  By the same calculation used for~\eqref{eq:I^2_1}, we obtain
  \begin{equation}
    \label{eq:I^2_3}
    I^2_3(s,t) \le C (t-s)^{2H}. \qedhere
  \end{equation}

\end{proof}

\begin{cor}
\label{cor:Holder}
Sample paths of the log-modulated fBm with Hurst parameter $H>0$ and $\zeta \le \frac{p}{1/2-H}$ are a.s.~$\alpha$-H\"{o}lder continuous for any $0<\alpha < H$.
\end{cor}

We continue to provide an explicit formula for the covariance between the log fBm  and a (correlated) Brownian motion.

\begin{lem}
  \label{lem:cov-log-fbm-bm}
  Let $\{B_t\}_{t\in 0,T]}$ be a standard Brownian motion correlated with the Brownian motion $\{W_t\}_{t\in 0,T]}$ driving the log-fBm $\{\widehat{W}_t\}_{t\in 0,T]}$ in \eqref{eq:log-fbm}, and let $\rho\in [-1,1]$ denote the correlation parameter. Denote $u
  \coloneqq t - t \wedge s$ and $v \coloneqq (u \vee \chi) \wedge t$. Then for $s,t\in [0,T]$, the covariance between $\widehat{W}_t$ and $B_s$ is given by 
  \begin{multline*}
    \cov\left( \hatW_t, \, B_s \right) = C \rho \biggl\{ \zeta^{-p}  \biggl[
      \log\left( \f{1}{v} \right)^{1-p} \Erm_p\left( (H+1/2) \log\left(
          \f{1}{v} \right) \right) \\
           - \mathds{1}_{u>0} \log\left( \f{1}{u}
      \right)^{1-p} \Erm_p\left( (H+1/2) \log\left( \f{1}{u} \right)
      \right)  \biggr] + \f{t^{H+1/2} - v^{H+1/2}}{H+1/2} \biggr\}.
  \end{multline*}
\end{lem}
\begin{proof}
  Direct computations reveal that
  \begin{multline*}
    \cov\left( \hatW_t, \, B_s \right) = E\left[ \hatW_t  B_s \right] = \rho
    \int_0^{t \wedge s} K(t-r) \dd r = \rho \int_u^t K(r) \dd r\\ = \rho
    \int_u^v K(r) \dd r + \rho \int_v^t K(r) \dd r
    = C \zeta^{-p} \rho \int_u^v r^{H-1/2} \log(1/r)^{-p} \dd r + C \rho
    \int_v^t r^{H-1/2} \dd r.
  \end{multline*}
  For the first integral, we use Lemma~\ref{lem:exp-integral} with $a = 1/2 -
  H$ and $b = p$, to obtain
  \begin{multline}\label{eq:general integral computation}
    \int_u^v r^{H-1/2} \log(1/r)^{-p} \dd r =  \bigg[
      \log\left( \f{1}{v} \right)^{1-p} \Erm_p\left( (H+1/2) \log\left(
          \f{1}{v} \right) \right)\\
          - \log\left( \f{1}{u}
      \right)^{1-p} \Erm_p\left( (H+1/2) \log\left( \f{1}{u} \right)
      \right)  \bigg],
  \end{multline}
  the second integral is trivial.
\end{proof}

\section{Fukasawa's method }\label{sec:Fukasawas method}
We give a short introduction to the asymptotic expansion for stochastic volatility models outlined in  \cite{fukasawa2011}, adapted to the case of Gaussian noise  driving the asset price and the volatility. This simplifies certain computations and conditions, and thus the results have been slightly changed accordingly.

Let $(\Omega,\cF^n,\{\cF_t^n\}_{t\in [0,T]},\PP)$ be a filtered probability space for each $n\in \NN$, where a continuous martingale $X^n$ lives. Consider an asset price process $S^n:[0,T]\times \Omega \rightarrow \RR^d$  given  by 
\begin{equation}\label{Z rep}
\begin{aligned}
S_t^n&=\exp(Z^n_t)
\\
Z_t^n &=Z_0+R(t)+A_t^n+X_t^n +\int_0^tg_s^n\dd W_s. 
\end{aligned}
\end{equation}
Here $\{g_t^n\}_{t\in [0,T]}$ is a process adapted to the filtration $\{\cF_t^n\}$. 
The Brownian motion $\{W_t\}_{t\in [0,T]}$ is independent of the martingale $X^n$,  and is correlated with the stochastic process $t\mapsto g_t^n$, in order to capture  the leverage effect. The function $R$ is supposed to reflect the interest rate and is often assumed to be constant, and in applications typically chosen to be zero.   $A^n$ is a drift term, such that $S^n$ is a martingale. 
Denote by $M_t^n$ the martingale part of $Z^n$, i.e. $M_t^n =X_t^n+\int_0^t g_s^n\dd W_s$. It is readily seen that the quadratic variation of $M^n$ is given by 
\begin{equation*}
\langle M^n\rangle =\langle X^n\rangle +\int_0^\cdot |g_s^n|^2 \dd s.
\end{equation*}
Throughout the text we will refer to $(R,A^n,X^n,g^n)$ as a stochastic volatility model. 
In applications, we will assume that the following hypothesis holds for the model $(R,A^n,X^n,g^n)$: 

\begin{hyp}\label{hypothesis 1}
For any null sequence $\{\epsilon_n\}$, there exists a sequence $\Sigma_n$ with $\Sigma=\lim_{n\rightarrow \infty} \Sigma_n>0$ such that for all $n\in \NN$ 
\begin{equation*}
\cD^n:=\epsilon_n^{-1}(\Sigma_n^{-1}\langle M^n \rangle_T-1)\qquad {\rm and} \qquad \frac{1}{\int_0^t |g_s^n|^2\dd s}
\end{equation*}
are bounded in $L^p(\Omega)$ for any $p>0$ Moreover,  $\cD^n$ and $\Sigma_n^{-\frac{1}{2}}M^n_T$ converges weakly to a random variable, say $(N_1,N_2)$. 
\end{hyp}

Fukasawa derives expansions of claims constructed from the model $(R,A^n,X^n,g^n)$, which can be written as $\EE[F(Z_T)]$, where $F(z)=e^{-R(T)}f(S_0\exp(z))$. 
A particular case of interest for the current article  is when the martingale part $M^n$ satisfies Hypothesis \ref{hypothesis 1} and the random variables $N_1$ and $N_2$ are normally distributed. In this case, the martingale expansion can be used to give an asymptotic expansion of the implied volatility in terms of the vol-of-vol parameter. 
The following theorem is a combination of Theorem 2.4 and Corollary 2.6 found in  \cite{fukasawa2011}. 

\begin{thm}\label{thm:implied VOl}
Suppose $F$ is a Borel measurable function of polynomial growth, and that Hypothesis \ref{hypothesis 1} holds with $N_1$ and $N_2$ being normally distributed. Denote by $\sigma_n^2T=\Sigma_n$.  Then the Black-Scholes implied volatility of European put/call options
can be expanded as
\begin{equation}\label{eq:IV}
\sigma_{BS}=\sigma_n\left(1+\frac{\epsilon_n}{2}\left(\delta-\rho_{12} d_2\right)\right)+o(\epsilon_n),
\end{equation}
where 
\begin{equation*}
\delta \coloneqq \EE[N_1],\qquad \rho_{12} \coloneqq \EE[N_1N_2],
\end{equation*}
 and 
\begin{equation*}
    d_1 \coloneqq \frac{\log(S/K)+r+\Sigma_n/2}{\sqrt{\Sigma_n}},\qquad d_2 \coloneqq d_1-\sqrt{\Sigma_n}.
\end{equation*}
\end{thm}

The following example is an application of the above theorem  to the case when the volatility is assumed to be driven by a Gaussian Volterra process. This particular example will motivate the subsequent discussions on volatility models driven by super-rough processes. As this example is essentially \cite[Sec.~3.3]{fukasawa2011} adapted to general Volterra processes, we will sometimes refer to this particular case as \emph{Fukasawa's example}.

\begin{ex}[Fukasawa's example with volatility driven by Gaussian Volterra processes]\label{Fuk Ex}
Suppose  $g$  is twice differentiable, positive function, with derivatives bounded away from zero. Consider the asset price dynamics given by 
\begin{equation*}
\begin{aligned}
S_t&=S_0\exp(Z_t)
\\
Z_t &= R(t)-\frac{1}{2} \int_0^tg(Y^n_s)^2\dd s + \int_0^t g(Y^n_s)[\rho \dd W_s+\sqrt{1-\rho^2}\dd W_s'], 
\end{aligned}
\end{equation*}
where for a null sequence $\{\epsilon_n\}$  we specify 
\begin{equation*}
Y_s^n = y+\epsilon_n \hatW_s,\qquad \hatW_t:=\int_0^t K(t-s)\dd W_s,
\end{equation*}
and $K$ is a square integrable, but possibly singular Volterra kernel, and
the two processes $\{W_t\}_{t\in [0,T]}$ and $\{W'_t\}_{t\in [0,T]}$ are independent Brownian motions. The parameter $\rho\in[-1,1]$ is the coefficient determining the correlation between $Y^n$ and $Z$. Observe also that $(\int_0^T \hatW_s\dd s,W_T)$ is normally distributed with 
\begin{equation*}
\EE[\int_0^T \hatW_s\dd s\, W_T]=\int_0^T \int_0^tK(t-s)\dd s\dd t.
\end{equation*}
Referring to \eqref{Z rep}, we now have $X^n=\int_0^\cdot g(Y^n_s)\sqrt{1-\rho^2}\dd W'_s$ and $A^n=-\frac{1}{2}\int_0^\cdot g(Y^n_s)^2 \dd s$.  Furthermore, $M^n_t := \int_0^t g(Y^n_s)[\rho \dd W_s+\sqrt{1-\rho^2}\dd W'_s]$, and we see that 
\begin{equation*}
\langle M^n\rangle_T = \int_0^T g(Y^n_s)^2\dd s.
\end{equation*} 
Set $\Sigma_n=\Sigma=g(y)^2 T$. Invoking the assumption  that $g$ is bounded away from $0$,  we see that for $n\in \NN$
\begin{equation*}
\epsilon_n^{-1}(\frac{\langle M^n\rangle_T}{g(y)^2T}-1) \qquad {\rm and} \qquad \frac{1}{\int_0^Tg^2(Y^n_s)\dd s} 
\end{equation*}
are both bounded in $L^p(\Omega)$ independently of $n$ for any $p>0$. Furthermore $\Sigma_n^{-\frac{1}{2}} M_T^n$ and $ \epsilon_n^{-1}(\frac{\langle M^n\rangle_T}{g(y)^2T}-1)$ are both seen to be asymptotically normally distributed, and thus the conditions in Hypothesis \ref{hypothesis 1} are satisfied. 
It follows directly  that $\delta=\EE[N_1]=0$. For the term $\rho_{12}$ from Theorem \ref{thm:implied VOl},  a second order Taylor expansion of $g$ around the point $y$ yields that 
\begin{equation*}
    g(Y^n_s)=g(y)+g'(y)\epsilon_n\widehat{W}_s +\int_0^1 g''(y +\theta\epsilon_n \widehat{W}_s)\dd \theta \, \epsilon_n^2 (\widehat{W_s})^2,
\end{equation*}
and thus using that $g$ is twice differentibale with bounded derivatives away from zero, and the independence of $W$ and $W'$ we have that 
\begin{equation}\label{rho_12 term in ex}
\rho_{12}(T) =\EE[N_1N_3]= \frac{g^\prime(y)\rho}{g(y)T^{3/2}}\int_0^T\EE[\hatW_s W_T]\dd s= \frac{g^\prime(y)\rho}{g(y)T^{3/2}}\int_0^T\int_0^t K(t-s)\dd s \dd t.
\end{equation} 
It follows from Theorem \ref{thm:implied VOl} that the Black-Scholes implied volatility  is given by 
\begin{equation}\label{eq:example IV}
\sigma_{BS} = g(y)\left(1-\frac{\epsilon_n\rho_{12}}{2} d_2\right), 
\end{equation}

In subsequent sections, we will investigate this term in more detail for the particular choice of the Volterra kernel $K$ given in \eqref{eq:log-kernel}. 
\end{ex}

\section{Skew expansions with log-fractional Brownian motion}
\label{sec:skew-expansion}

We now apply the small vol-of-vol expansion to log-modulated rough volatility models. In the first step, we compute the term $\rho_{12}$ for such models.

\subsection{$\bm{\rho_{12}}$ in the case of log-fractional Brownian motion}
\label{sec:fukas-term-rho_12}

We will compute the term $\rho_{12}$ given in Example \ref{Fuk Ex} when the Volterra process is given as a log-fractional Brownian motion.  Recall from \eqref{rho_12 term in ex} that $\rho_{12}$ is given by 
\begin{equation}
  \label{eq:rho-13}
  \rho_{12} = \f{g^\prime(y) \rho}{g(y) T^{3/2}} \int_0^T E\left[ \hatW_s W_T 
  \right] \dd s , 
\end{equation}
where $\rho$ is the correlation coefficient between the Brownian noises. 
We compute the integral on the r.h.s.~under the assumption that $T$ is small,
more precisely, $T \le \chi$. (Keep in mind that we are eventually going to look for
asymptotics for $T \to 0$.)

\begin{lem}
  \label{lem:integrated-covariance}
 Let $T>0$ and $\{W_t\}_{t\in [0,T]}$ be a Brownian motion, and define the log-fractional Brownian motion  $\widehat{W}_t=\int_0^t K(t-s)\dd W_s$, where the kernel is given as in \eqref{eq:log-kernel}. Then we have 
  \begin{equation*}
    \int_0^T \EE\left[ W_T \hatW_s \right] \dd s = \int_0^T \int_0^s K(s-r) \dd r \dd s = I_1(T \wedge \chi) + \mathbbm{1}_{T>\chi} \left( I_2(T,\chi) + I_3(T, \chi) \right),
  \end{equation*}
  where
  \begin{align*}
      I_1(T) &\coloneqq C \zeta^{-p} \log\left(\f{1}{T} \right)^{1-p} \biggl[ T
     \Erm_p\left( (H+1/2) \log\left(
        \f{1}{T} \right) \right)  - \Erm_p\left( (H+3/2) \log\left( \f{1}{T} \right)
    \right) \biggr],\\
    I_2(T, \chi) &\coloneqq \f{C}{H+1/2} \left( \f{T^{H+3/2} - \chi^{H+3/2}}{H+3/2} - (T - \chi)\chi^{H+1/2} \right),\\
    I_3(T, \chi) &\coloneqq C \zeta^{-p} (T - \chi)  \log\left( \f{1}{\chi} \right)^{1-p} \Erm_p\left( (H+1/2) \log\left(\f{1}{\chi}\right)\right).
  \end{align*}
\end{lem}
\begin{proof}
  It is tempting to integrate the formula in Lemma~\ref{lem:cov-log-fbm-bm},
  but we were not able to find a closed form expression this way. Rather, let
  us start from scratch. Clearly, we have that
  \begin{equation*}
      \int_0^T\EE[W_T\widehat{W}_s]\dd s = \int_0^T \int_0^s K(s-r)\, \dd r\, \dd s \eqqcolon I(T).
  \end{equation*}
  
  We first assume that $T \le \chi$.
 Using the representation of the kernel given in \eqref{eq:log-kernel} we see that
  \begin{align*}
   \int_0^{T} \int_0^s K(s-r)\dd r\dd s 
    &= \int_0^{T} C \zeta^{-p} \int_0^s (s-r)^{H-1/2} \log\left( \f{1}{s-r}
      \right)^{-p} \dd r \, \dd s\\
    &= C \zeta^{-p} \int_0^{T} \int_0^s r^{H-1/2} \log\left( \f{1}{r}
      \right)^{-p} \dd r \, \dd s\\
    &= C \zeta^{-p} \int_0^{T} \int_r^{T} \dd s \ r^{H-1/2} \log\left( \f{1}{r}
      \right)^{-p} \dd r \\
    &= C \zeta^{-p} \left[ {T} \int_0^{T} r^{H-1/2} \log\left( \f{1}{r}
      \right)^{-p} \dd r - \int_0^{T} r^{H+1/2} \log\left( \f{1}{r}
      \right)^{-p} \dd r \right].
  \end{align*}
  By Lemma~\ref{lem:exp-integral} it then follows that
  \begin{gather*}
    \int_0^{T} r^{H-1/2} \log\left( \f{1}{r}
    \right)^{-p} \dd r = 
    \log\left(\f{1}{{T}} \right)^{1-p} \Erm_p\left( (H+1/2) \log\left(
        \f{1}{{T}} \right) \right),\\
    \int_0^{T} r^{H+1/2} \log\left( \f{1}{r}
    \right)^{-p} \dd r =  \log\left(\f{1}{{T}}
    \right)^{1-p} \Erm_p\left( (H+3/2) \log\left( \f{1}{{T}} \right)
    \right). 
  \end{gather*}
 Putting the terms together, we obtain
 \begin{multline*}
     I(T) = I_1(T) =  C \zeta^{-p} \log\left(\f{1}{T} \right)^{1-p} \biggl[ T
     \Erm_p\left( (H+1/2) \log\left(
        \f{1}{T} \right) \right) \\
        -\Erm_p\left( (H+3/2) \log\left( \f{1}{T} \right)
    \right) \biggr].
 \end{multline*}
 
 Let us now consider the case $T > \chi$. The integral can then naturally be split as
 \begin{align*}
     I(T) &= \int_0^{\chi} \int_0^s K(s-r) \dd r \dd s + \int_\chi^T \int_0^{s-\chi} K(s-r) \dd r \dd s 
     + \int_\chi^T \int_{s-\chi}^s K(s-r) \dd r \dd s\\
     &= I(\chi) + \int_\chi^T \int_\chi^s K(u) \dd u \dd s + \int_\chi^T \int_0^\chi K(u) \dd u \dd s,
 \end{align*}
 noting that $I(\chi)$ is already known. An elementary calculation gives us the second term,
 \begin{multline*}
     \int_\chi^T \int_\chi^s K(u)\, \dd u\, \dd s = C \int_\chi^T \int_\chi^s u^{H-1/2}\, \dd u\, \dd s = \f{C}{H+1/2} \int_\chi^T \left( s^{H+1/2} - \chi^{H+1/2} \right) \dd s \\
     = \f{C}{H+1/2} \left( \f{T^{H+3/2} - \chi^{H+3/2}}{H+3/2} - (T - \chi)\chi^{H+1/2} \right).
 \end{multline*}
 Finally, regarding the third term we do a substitution of variables and apply Lemma~\ref{lem:exp-integral} to obtain 
 \begin{align*}
     \int_\chi^T \int_0^\chi K(u) \dd u \dd s &= C \zeta^{-p} (T - \chi) \int_0^\chi \log\left( \f{1}{u} \right)^{-p} u^{H-1/2} \dd u\\ 
     &= C \zeta^{-p} (T - \chi)  \log\left( \f{1}{\chi} \right)^{1-p} \Erm_p\left( (H+1/2) \log\left(\f{1}{\chi}\right)\right). \qedhere
 \end{align*}
\end{proof}

\subsection{Asymptotic expansion for Example~\ref{Fuk Ex}}

We continue with a discussion of Example~\ref{Fuk Ex}, when the volatility depends on a log-fBm. As we have already computed $\rho_{12}$, we have all ingredients for the asymptotic expansion in terms of small vol-of-vol. We are also interested in the short time behaviour of this term, which relies on the following well known asymptotic expansion of the exponential integral $\Erm_p$ :
\begin{equation}
  \label{eq:exp-int-expansion}
  \Erm_p(x) \sim \frac{e^{-x}}{x}\left[1 - \frac{p}{x} + \frac{p(p+1)}{x^2} \pm
  \cdots \right] \text{ as } x \to \infty.
\end{equation}

\begin{lem}
  \label{lem:rho13-expansion}
  Let $g$ be a positive  twice continuously differentiable function with derivatives bounded away from $0$. The term $\rho_{12}$ in \eqref{eq:rho-13} satisfies the asymptotic expansion
  \begin{equation*}
    \rho_{12} = \f{g^\prime(y)}{g(y)}\frac{C \zeta^{-p}\rho}{(H+1/2)(H+3/2)} 
    \log\left(\f{1}{T} \right)^{-p} T^{H}(1+o(1))
  \end{equation*}
  as $T \to 0$.
\end{lem}
\begin{proof}
  Using the asymptotic expansion in~\eqref{eq:exp-int-expansion} we have
  \begin{gather}\label{eq:asym:Ep1}
    \Erm_p\left( (H+1/2) \log\left( \f{1}{T} \right) \right) = \frac{T^{H+1/2}}{(H+1/2) \log\left(
        \frac{1}{T} \right)} (1+o(1)),\\
        \label{eq:asym:Ep2}
    \Erm_p\left( (H+3/2) \log\left( \f{1}{T} \right) \right) = \frac{
        T^{H+3/2}}{(H+3/2) \log\left(
        \frac{1}{T} \right)} (1+o(1)),
  \end{gather}
  as $T \to 0$. Hence (as $T < \chi$ eventually), we obtain 
  \begin{equation*}
    \int_0^T \EE\left[ W_T \hatW_s \right] \dd s = \frac{C \zeta^{-p}
      }{(H+1/2)(H+3/2)} 
    \log\left(\f{1}{T} \right)^{-p} T^{H+3/2}(1+o(1)).
  \end{equation*}
  Recalling~\eqref{eq:rho-13}, we have
  \begin{equation*}
    \rho_{12} = \f{g^\prime(y) \rho}{g(y) T^{3/2}} \int_0^T E\left[ W_T \hatW_s
    \right] \dd s = \f{g^\prime(y)}{g(y)}\frac{C \zeta^{-p}\rho}{(H+1/2)(H+3/2)} 
    \log\left(\f{1}{T} \right)^{-p} T^{H}(1+o(1)). \qedhere
  \end{equation*}
\end{proof}

Let us now look again at the implied volatility found in Theorem  \ref{thm:implied VOl} and in Example \ref{Fuk Ex}. By formula \eqref{eq:IV} we have
\begin{equation*}
  \sigma_{BS} = \sigma\left(1 - \frac{\rho_{12}d_2}{2} \epsilon_n \right)
  (1+o(\epsilon_n)),
\end{equation*}
Following Example \ref{Fuk Ex}, we set $R(t)\coloneqq rt$ for a constant $r>0$,   $\sigma \coloneqq g(y)$,  $\Sigma \coloneqq g(y)^2 T$, and
\begin{equation*}
  d_2 = \frac{\log\left( \frac{S_0}{K} \right) + r -
    \frac{\Sigma}{2}}{\sqrt{\Sigma}}. 
\end{equation*}
Hence, the part of the leading order term depending on log-moneyness
$\log\left( \frac{S_0}{K} \right)$ is
\begin{equation*}
  -\frac{\sigma}{2} \rho_{12} \frac{\log\left( \frac{S_0}{K}
    \right)}{g(y)\sqrt{T}} \epsilon_n=  - a \rho
  \log\left(\f{1}{T} \right)^{-p} T^{H-1/2} \log\left( \frac{S_0}{K}
  \right) \epsilon_n (1+o_T(1)),
\end{equation*}
with
\begin{align}\label{eq:a}
  a &\coloneqq \frac{1}{2} \f{g^\prime(y)}{g(y)}\frac{C \zeta^{-p}}{(H+1/2)(H+3/2)} \nonumber\\
  &= 
  \begin{cases}
  \frac{1}{2} \f{g^\prime(y)}{g(y)} \left[ \sqrt{\Erm(2H/\zeta)/\zeta + \frac{1-\exp(-2H/\zeta)}{2H}} \zeta^p (H+1/2) (H+3/2) \right]^{-1}, & H > 0,\\
  \frac{3}{8} \f{g^\prime(y)}{g(y)} \zeta^{1/2-p} \sqrt{\frac{2p-1}{2p}}, & H = 0.
  \end{cases}
\end{align}

\begin{rem}
    \label{rem:zeta-p-vol-of-vol}
    As the skew asymptotic is linear in $a = a_{H,\zeta,p}$, we may think of these model parameters to contribute to \emph{vol-of-vol}. It turns out that $a$ varies considerably as a function of $\zeta$ and $p$ for fixed roughness $H$. The actual asymptotic skew formula is, fortunately, much more stable, see Figure~\ref{fig:sberg_p_zeta}.
\end{rem}

These considerations lead to the following theorem regarding the ATM volatility skew for small vol-of-vol and short maturity $T$: 
\begin{thm}\label{thr:skew}
  The implied volatility in Example~\ref{Fuk Ex} with log-modulated fBm satisfies
  \begin{equation*}
      \sigma_{BS} = g(y) \left(1 - \frac{\rho_{12}d_2}{2} \epsilon_n \right) (1+o(\epsilon_n))
  \end{equation*}
  with $\rho_{12}$ given by~\eqref{eq:rho-13} together with Lemma~\ref{lem:integrated-covariance} and
  \begin{equation*}
    d_2 = \frac{\log\left( \frac{S_0}{K} \right) + r - \frac{g(y)^2T}{2}}{g(y) \sqrt{T}}. 
  \end{equation*}
  For log-moneyness $k, k'\in \RR$,  short maturity $T$, and any $0 \le H \le 1/2$, the skew therefore behaves like
  \begin{equation*}
    \frac{\sigma_{BS}(T,k)-\sigma_{BS}(T,k')}{k-k'} \approx - a \rho \log\left(\f{1}{T} \right)^{-p} T^{H-1/2} \epsilon_n,
  \end{equation*}
  with $a$ defined in \eqref{eq:a}.
\end{thm}

\subsection{The rough Bergomi model}

As a practical example, we consider here the rough Bergomi model, when the driving noise of the instantaneous variance is given as a log-fBm.
To this end, denote by $\xi(u)=\EE_{\QQ}[v_u|\cF_0]$,  where $v$ denotes the instantaneous variance. 
The rough Bergomi model is given by 
\begin{equation}\label{eq: def of rBergomi dyn}
\begin{aligned}
S_t^n&=S_0\cE\left(\int_0^t \sqrt{v_s^n}\dd B_s\right)
\\
v_t^n&=\xi(t)\cE\left(\epsilon_n \int_0^tK(t-s)\dd W_s\right)
\end{aligned}
\end{equation}
where $B_t=\rho W_t+\sqrt{1-\rho^2}W'_t$. 
\begin{thm}
\label{thr:rbergomi}
Let for $n\in \NN$, let  $(S^n,v^n)$ be a stochastic volatility model  given with rough Bergomi dynamics as  in \eqref{eq: def of rBergomi dyn}, where $\{\epsilon_n\}_{n\in \NN}$ is a null sequence, representing vol-of-vol $\eta$, and where the Volterra kernel $K$ is given as in \eqref{eq:log-kernel}. Then the following expansion holds for the implied volatility surface 
\begin{equation*}
\sigma_{BS}(T,k) = k\, \epsilon_n\, \frac{T^{-\frac{1}{2}}}{2}\left(\int_0^T \xi(s)\dd s\right)^{-\frac{3}{2}}\int_0^T \xi(s)\int_0^sK(s-r)\sqrt{\xi(r)}\dd r\dd s  + o(\epsilon_n), 
\end{equation*}
where $k=\log (S_0/K)$ denotes log-moneyness and $T$ is maturity time. 
Furthermore, the ATM volatility skew behaves like 
\begin{equation*}
     \frac{\sigma_{BS}(T,k)-\sigma_{BS}(T,k')}{k-k'} \approx \frac{\epsilon_n T^{-\frac{1}{2}}}{2}\left(\int_0^T \xi(s)\dd s\right)^{-\frac{3}{2}}\int_0^T \xi(s)\int_0^s K(s-r)\sqrt{\xi(r)}\dd r\dd s.
\end{equation*}
\end{thm}

\begin{proof}
For the proof of this theorem we will apply the martingale expansion of Theorem \ref{thm:implied VOl} to obtain the implied volatility expansion. 
To this end, we need to verify that Hypothesis \ref{hypothesis 1} holds for this particular model. Since $g(y)\sim e^{y}$ is unbounded, we cannot apply Lemma \ref{lem:rho13-expansion} directly, and we need to verify that  the conditions in Hypothesis \ref{hypothesis 1} indeed holds. 
We begin to specify the terms of Fukasawa's expansion.

In this case we have 
\begin{align*}
A^n_t&=-\frac{1}{2}\int_0^t v_s^n\dd s
\\
M_t^n &=\sqrt{1-\rho^2}\int_0^t \sqrt{v_s^n}\dd W'_s+\rho\int_0^t \sqrt{v_s^n}\dd W_s. 
\end{align*}
It is readily seen that $\langle M^n\rangle =- 2A^n$. 
We set $\Sigma_n=\Sigma=\int_0^T\xi(s)\dd s$ and observe that  for each $n\in \NN$
\begin{equation*}
\cD^n =\epsilon^{-1}_n(\Sigma^{-1}\int_0^T \xi(s)\cE\left(\epsilon_n \int_0^sK(s-r)\dd W_r\right)\dd s-1) 
\end{equation*}
is bounded in   $L^p(\Omega)$ for any $p>0$. Furthermore, by Jensen's inequality, it follows that $\left(\int_0^T v_s^n \dd s\right)^{-1}$  is bounded in $L^p(\Omega)$. Indeed, we see that
\begin{equation*}
\EE\left[ \left(\int_0^T v_s^n\dd s\right)^{-p}\right] \leq \int_0^T\EE\left[ (v_s^{n})^{-p}\right]\dd s <\infty,
\end{equation*}
where we have used that $(v_s^n)^{-1}=v_t^n=\xi(t)\cE\left(-\epsilon_n \int_0^tK(t-s)\dd W_s\right)$ which is contained in $L^p(\Omega)$.
 Moreover, we see that $\cD^n$ and  $\Sigma^{-1/2}M_T^n$ converge weakly to the normal random variables $N_1$ and $N_2$. 
In particular, we have that 
\begin{equation*}
N_1=\frac{\int_0^T\xi(s)\int_0^t K(s-r)\dd W_r \dd s}{\Sigma}
\qquad {\rm and} \qquad 
N_2 = \frac{ \int_0^T \sqrt{\xi(s)}\dd B_s  }{\sqrt{\Sigma}}.
\end{equation*}
We can therefore apply Theorem \ref{thm:implied VOl} to the rough Bergomi model. To this end, we need to compute $\rho_{12}=\EE[N_1N_2]$, and we observe that 
\begin{equation*}
\rho_{12}=\rho\Sigma^{-\frac{3}{2}}\int_0^T \xi(s) \int_0^s K(s-r) \sqrt{\xi(r)} \dd r \dd s. 
\end{equation*}
Explicit computations of this term is more difficult, due to the integration over the variance curve. Of course, if $\xi(s)=\xi$ is constant, then $\rho_{12}$ is computed identically as in Lemma \ref{lem:rho13-expansion}. 
 It follows from Theorem \ref{thm:implied VOl}, using that $\sigma_n=\sqrt{\frac{\Sigma}{T}}$,  that the implied volatility is given by 
\begin{equation*}
\sigma_{BS} = \sqrt{\frac{\Sigma}{T}}\left(1-\frac{\epsilon_n}{2}\rho_{12}d_2   \right)+o(\epsilon_n)
\end{equation*}
Inserting the values for $d_2$ and $\rho_{12}$, considering the leading order term involving the log-moneyness $k=\log\left(\frac{S}{K}\right)$, we find that 
\begin{equation*}
\sigma_{BS}(T,k) = k \frac{\epsilon_nT^{-\frac{1}{2}}}{2}\Sigma^{-\frac{3}{2}}\int_0^T \xi(s)\int_0^sK(s-r)\sqrt{\xi(r)}\dd r\dd s + o(\epsilon_n), 
\end{equation*}
where $k=\log \left(\frac{S}{K}\right)$. 
Furthermore, from the above formula, it is straightforward to see that the ATM volatility skew behaves like
\begin{equation*}
   \frac{\epsilon_nT^{-\frac{1}{2}}}{2}\Sigma^{-\frac{3}{2}}\int_0^T \xi(s)\int_0^s K(s-r)\sqrt{\xi(r)}\dd r\dd s. 
\end{equation*}
Substituting $\Sigma=\int_0^T \xi(s)\dd s$, concludes the proof. 
\end{proof}

\begin{cor}
\label{cor:rbergomi-short-skew}
In the super-rough Bergomi model with constant forward variance curve $\xi$, the ATM skew behaves like 
\begin{equation*}
    \frac{\sigma_{BS}(T,k)-\sigma_{BS}(T,k')}{k-k'} \approx  - a \rho \log\left(\f{1}{T} \right)^{-p} T^{H-1/2} \epsilon_n,\qquad {\rm as} \quad T\rightarrow 0,
  \end{equation*}
  with $a$ given in \eqref{eq:a}, substituting $g(x) \equiv \xi(0) \exp(x)$.
\end{cor}

In the formula above as well as in Theorem \ref{thr:rbergomi}, we write ``$\approx$'' in an informal way, meaning that we compute an approximation of the finite difference corresponding to the skew using our small vol of vol expansion. Then, ``$\approx$'' stands for the behavior as $T\to 0$ of this quantity.

Corollary~\ref{cor:rbergomi-short-skew} provides a financial interpretation of the parameter $p$ as a log-modulation of the power-law behaviour of the ATM-skew for short maturities. We hesitate to provide any financial interpretation to the second parameter $\zeta$ of the log-modulated fBm. This is in line with Remark~\ref{rem:parameter-choice}, which recommends fixing $\zeta$ a priori rather than calibrating it to financial data.

\section{Asymptotic skew under log-fractional volatility}\label{sec: asymptotic paolo}

In Section \ref{sec:skew-expansion} we show an expansion for the implied skew in small time and small vol-of-vol using the martingale expansion outlined in Section \ref{sec:Fukasawas method}. We attempt here to understand the short time behavior of a log-modulated rough stochastic volatility model without considering the small vol-of-vol regime, but just the short time asymptotics. For this, we adapt Fukasawa's framework of \cite{fukasawa2017,euch2019short, fukasawa2020} to log-fractional volatility, using the ``regular variation'' language.
Let us consider, similarly to \eqref{eq: def of rBergomi dyn}, a stochastic volatility model of the form
\begin{equation*}
\begin{split}
 &
 S_t = 
 S_0\, \mathcal{E}\left(\int_0^t\sqrt{v_s}
\dd B_s\right), 
\\&v_t = 
\xi(t) \,\mathcal{E}
\left(
\eta\int_0^t K(t,s) \dd W_s 
\right)
\end{split}
\end{equation*}
where  $B_s=\rho  W_s
+ \sqrt{1-\rho^2}  W'_s$, $\eta>0$. For now, we do not assume the specific form \eqref{eq:log-fbm} for $K$, but only square integrability,  in the sense that $K(t,\cdot)\in L^2([0,t])$ for all $t\in [0,T]$. 
We also assume
that the process $\int_0^t K(t,s) \dd W_s$  is continous (as already seen, this is satisfied for $K$ in \eqref{eq:RL-kernel}; see also \cite[Proposition 2.4]{gassiat} for a general continuity condition for convolution kernels) and that $\rho \in [-1,0]$. This implies that the price process $S$ is a martingale, as shown in 
\cite[Theorem 1.1]{gassiat}.
Let $s(t):=\left(\int_0^t K(t,s)^2 \dd s \right)^{1/2}$. To allow logarithmic corrections to the fractional power-law  type kernels we assume $s(t)\to 0$ as $t \to 0$ and $s$ to be regularly varying  at $0$: for some $L$ slowly varying,
\[
s(t)=t^H L(t)
\] 
(so $L(t)\to 0$ if $H=0$).
We also assume $H\in[0,1/2)$ (rough but also super-rough volatility).
Let
\begin{equation*}
\begin{split}
& \bar{\xi}(t) \coloneqq \frac{1}{t} \int_0^t \xi(u) \dd u, 
\\
& \cK
\coloneqq
\lim_{t\to 0}
\frac{\int_0^t K(t,s) \dd s}{\sqrt{t\int_0^t K^2(t,s) \dd s}}
\\
& \alpha(z) \coloneqq z \frac{\rho \eta \cK}{\sqrt{v_0}(2H+3)} 
\end{split}
\end{equation*}
where $v_0$ is spot volatility and $\xi(\cdot)$ is continuous at $0$. The following theorem and corollary are inspired by \cite[Theorem 2.1 and Corollary 2.1]{fukasawa2020}, modified in order to be applicable to log-fractional volatility.
\begin{thm}
Denoting $\sigma_{{BS}}(k,T)$ the Black-Scholes
 implied volatility at time $0$ with expiry $T$ and log-moneyness $k$. For $z\in \RR$ and $T \to 0$,
\begin{equation*}
 \sigma_{{BS}}\left(z\sqrt{T},T\right) = 
\sqrt{\bar{\xi}(T)} \big(1 + \alpha(z){s}(T)\big)
+  o({s}(T))
\end{equation*}
\end{thm}

\begin{cor}\label{cor:skew:smalltime}
The implied skew behaves as follows: for $z' \neq z$, if $\cK\neq 0$,
\begin{equation*}
\begin{split} 
 \frac{ \sigma_{{BS}}\left(z\sqrt{T},T\right) -
 \sigma_{{BS}}\left(z' \sqrt{T},T\right)}
{z\sqrt{T} - z' \sqrt{T} }
&\sim 
\frac{\rho \eta \cK}{2H+3} 
\frac{ {s}(T) }{\sqrt{ T}}
,
\end{split}
\end{equation*}
where $\sim$ denotes asymptotic equivalence as $T\to 0$. If $\cK=0$,
\begin{equation*}
\begin{split} 
\frac{\sqrt{ T}}{ {s}(T) } \frac{ \sigma_{{BS}}\left(z\sqrt{T},T\right) -
 \sigma_{{BS}}\left(z' \sqrt{T},T\right)}
{z\sqrt{T} - z' \sqrt{T} }
&\to 0.
\end{split}
\end{equation*}
\end{cor}
\begin{rem}
It always holds $\cK\leq 1$.
This corollary gives the exact scaling of the implied skew if $\cK\neq 0$, otherwise just gives an upper bound. 
\end{rem}

\begin{proof}
This proof is based on \cite[Appendix A]{fukasawa2020}, \cite[Theorem 1]{fukasawa2017}. Combining 
the martingale CLT (see e.g. \cite[Chapter VIII]{JacodShiryaev}), localization arguments and explicit computations with log-normal random variables, we get
 \begin{equation*}
\left(\frac{1}{\sqrt{t}}  \left(\frac{S_t}{S_0}-1\right)
-
\frac{\sqrt{v_0}B_t}{\sqrt{t}}
,
\frac{1}{\eta s(t)}
\left(
\frac{v_t}{\xi(t)} -1 \right)
-
\frac{\int_0^t K(t,s)\dd W_s}{ s(t)}
\right)
\to
(0,0)
 \end{equation*}
in law as
 $t \to 0$.
Therefore
\begin{equation*}
\left(\frac{1}{\sqrt{t}}  \left(\frac{S_t}{S_0}-1\right),
\frac{1}{\eta s(t)}
\left(
\frac{v_t}{\xi(t)} -1 \right)
\right)
\to
(\gamma,\delta)
 \end{equation*}
in law as
 $t \to 0$,
where $(\gamma,\delta)$
is a centred 2-dim Gaussian with covariance
\begin{equation}\label{eq:gaussian_cov}
\Sigma=\begin{pmatrix}
v_0 & \rho \sqrt{v_0} 
\cK
 \\
 \rho \sqrt{v_0} 
\cK & 1
\end{pmatrix}.
\end{equation}
For $t>0,u \in [0,1]$ let us write
\begin{equation*}
 X_u^t = \frac{1}{\sqrt{t}} 
\left(\frac{S_{ut}}{S_0}-1 \right)
\end{equation*}
From the fact that $S$ is a martingale, it follows that $X^t$ is a martingale in $u$ for fixed $t$, with quadratic variation
\begin{equation*}
  d \langle X^t \rangle_u = (S_{tu}/S_0)^2 v_{ut}
   \dd u = 
(1 + \sqrt{t} X_u^t)^2 v_{ut}  \dd u.
\end{equation*}
We write $\Delta = (e^{z
  \sqrt{t}} -1)/\sqrt{t}$.
We use the Bachelier pricing equation as in \cite{fukasawa2017,fukasawa2020},
\begin{equation*}
 \frac{\partial p}{\partial u} (x,u)
+ \frac{1}{2}\xi(ut) \frac{\partial^2 p}{\partial x^2}(x,u) = 0,\quad\quad
p(x,1) = (\Delta - x)^+
\end{equation*}
whose explicit solution is given by
\begin{equation}\label{p:explicit}
p(x,u) = (\Delta-x) \Phi\left(\frac{t(\Delta-x)}{\int_{ut}^t \xi(s) \dd s}\right)
+
\phi\left(\frac{t(\Delta-x)}{\int_{ut}^t \xi(s) \dd s}\right)
\,\frac{1}{t}
\int_{ut}^t \xi(s) \dd s
\end{equation}
with $\Phi,\phi$ standard normal distribution function and density.
By It\^{o}'s formula we rewrite the following rescaled put option price in terms of $X^t$
\begin{equation}\label{eq:exp}
\begin{split}
\frac{E[(S_0e^{z \sqrt{t}}-S_t)^+]}{S_0
  \sqrt{t}}
 &=  E[(\Delta-X^t_1)^+] = E[p(X^t_1,1)] \\
&= p(0,0) +
\frac{1}{2}
E\left[\int_0^1
\frac{\partial^2 p}{\partial x^2}(X_u^t,u)
(
v_{ut} - \xi(ut))
 \dd u\right]
\\
&+
\frac{1}{2}
E\left[\int_0^1
\frac{\partial^2 p}{\partial x^2}(X_u^t,u)
(
2\sqrt{t} X_u^t+ {t} (X_u^t)^2)v_{ut}
 \dd u\right]
\end{split}
\end{equation}
Since (as in \cite{fukasawa2020})
\begin{equation*}
 \begin{split}
\int_0^1E\left[\frac{\partial^2 p}{\partial x^2}(X_u^t,u)
X_u^t v_{ut}\right] 
 \dd u 
 \to v_0
\frac{z\sqrt{v_0}}{2}\phi\left(\frac{z}{\sqrt{v_0}}\right)
 \end{split}
\end{equation*}
we have
\[
E\left[\int_0^1
\frac{\partial^2 p}{\partial x^2}(X_u^t,u)
(
2\sqrt{t} X_u^t+ {t} (X_u^t)^2)v_{ut}
 \dd u\right] = O(\sqrt{t})
\]
so this term is negligible in \eqref{eq:exp}. Now we use our different (possibly logarithmic) scaling assumption for the volatility and get
\begin{equation*}
 \left(X_u^t,   
\frac{v_{tu} - \xi(tu)}{\eta s(tu)}
\right)
\xrightarrow{t\to 0} (\sqrt{u}\gamma, v_0\delta).
\end{equation*}
Again as in \cite{fukasawa2020}, as $t \to 0$ we have
\begin{equation*}
 \frac{\partial^2 p}{\partial x^2}(X_u^t,u)
\xrightarrow{t\to 0} 
\frac{1}{\sqrt{v_0(1-u)}}\phi\left(
\frac{z-\sqrt{u}\gamma}{\sqrt{v_0(1-u)}}\right)
\end{equation*}
in law for each $ u \in [0,1)$. We have
\begin{equation*}
\frac{1}{s(t)} \int_0^1E\left[\frac{\partial^2 p}{\partial
  x^2}(X_u^t,u)
(v_{ut} - \xi(ut))\right]
 \dd u 
  =
\int_0^1
\frac{s( ut)}{s(t)  u^H
}
\frac{u^H}{s( ut)}
E\left[\frac{\partial^2 p}{\partial
  x^2}(X_u^t,u)
(v_{ut} - \xi(ut))\right]
 \dd u.
\end{equation*}
Regular variation of $s(\cdot)$ implies $s( ut) \sim u^H s(t)  $ as $t\to 0$.
So, 
\begin{equation*}
 \begin{split}
  & 
  \lim_t
  \frac{1}{
s(t)} \int_0^1E\left[\frac{\partial^2 p}{\partial
  x^2}(X_u^t,u)
(v_{ut} - \xi(ut))\right]
 \dd u \\
 & =   \lim_t
\int_0^1
\frac{u^H}{s( ut)}
E\left[\frac{\partial^2 p}{\partial
  x^2}(X_u^t,u)
(v_{ut} - \xi(ut))\right]
 \dd u \\
&  
=\int_0^1 
E\left[
\frac{u^H \eta v_0 \delta}{\sqrt{v_0(1-u)}}\phi\left(
\frac{z-\sqrt{u}\gamma}{\sqrt{v_0(1-u)}}\right)
\right] \dd u.
\end{split}
\end{equation*}
The joint (Gaussian) density of $\gamma$ and $\delta$ is given in \eqref{eq:gaussian_cov}. Explicit computations give
\[
  \lim_t
  \frac{1}{
s(t)} \int_0^1 E\left[\frac{\partial^2 p}{\partial
  x^2}(X_u^t,u)
(v_{ut} - \xi(ut))\right]
 \dd u =
 \frac{
z\rho\eta \cK}{H+3/2}
\phi\left(\frac{z}{\sqrt{v_0}}\right).
\]
Now, from the definition of $\alpha$ and \eqref{p:explicit}, we write the rescaled put option with expiry $t$ as
\begin{equation} \label{eq:exp_put}
 \begin{split}  \frac{E[(S_0e^{z \sqrt{t}}-S_t)^+]}{S_0
  \sqrt{t}}  & = p(0,0) +  
\alpha(z)\sqrt{v_0}\phi\left(\frac{z}{\sqrt{v_0}}\right)
 {s}(t)  + o({s}(t)) \\
&=
\Delta \Phi\left(\frac{\Delta}{\sqrt{\bar{\xi}(t)}}\right)
+ \sqrt{\bar{\xi}(t)}
\phi\left(\frac{\Delta}{\sqrt{\bar{\xi}(t)}}\right)
\left(1 + \alpha(z) {s}(t)
\right)
  + o({s}(t)).
 \end{split}
\end{equation}
Let $p_{{BS}}(K,t,\sigma)$ denote the price  under the Black-Scholes model of a put option with strike $K$, expiry $t$ and volatility $\sigma$. We have the following Taylor expansion, holding for fixed $a$, analogous to \cite[Equation (6)]{fukasawa2020}
\[
 \frac{p_{{BS}}(S_0e^{z\sqrt{t}},t, \sigma +
  a {s}(t) )}{S_0 \sqrt{t}}
=
\Delta \Phi\left(\frac{\Delta}{\sigma}\right)
+ \sigma
\phi\left(\frac{\Delta}{\sigma}\right)
\left(1 
+ \frac{a}{\sigma} 
{s}(t) \right)
  + o({s}(t)).
\]
We have equality with \eqref{eq:exp_put} with
\[
\sigma = \sqrt{\bar{\xi}(t)}, \quad
 a =  \alpha(z) \sqrt{\bar{\xi}(t)},
\]
and the implied volatility expansion follows taking $T=t$ (at least formally).
\end{proof}

\subsection{Asymptotic skew of the (super) rough Bergomi model}
We consider now the model
with $K$ given in  \eqref{eq:log-fbm} and \eqref{eq:log-kernel}. We have, using Lemma \ref{lem:exp-integral}, Lemma \ref{lem:var-logfbm} and \eqref{eq:asym:Ep1}
\[
\begin{split}
\int_0^{t}K(t,s)^2 d s &\sim
\begin{cases}
\frac{C^2 \zeta^{-2p}}{2p-1} \log(1/t)^{1-2p}
\mbox{ for } H=0
\\
\frac{C^2 \zeta^{-2p}}{2H} \log(1/t)^{-2p} t^{2H}
\mbox{ for } H>0
\end{cases} \\
\int_0^{t}K(t,s) d s
&\sim
\frac{C \zeta^{-p}}{H+1/2} \log(1/t)^{-p} t^{H+1/2} \mbox{ for } H\geq 0
\end{split}
\]
We get $\cK=\sqrt{2H}/(H+1/2)$ for $H\in[0,1/2)$.
So, writing ``skew''
in the sense of Corollary \ref{cor:skew:smalltime},
\[
 \text{skew}  \sim
\frac{\rho \eta C\zeta^{-p}}{(2H+3)(H+1/2)}T^{H-1/2} \log(1/T)^{-p},
\]
for $H>0$, and we recover the analogous result to \eqref{eq:skew-asymptotic} and Theorem \ref{thr:rbergomi}.
For $H=0$, we can say
\begin{equation*}
T^{1/2}
\log(1/T)^{p-1/2}
 \text{skew}   \to 0,
\end{equation*}
which gives an upper bound, but we do not get the precise time-scaling of the skew. However, this upper bound is consistent with the small vol-of-vol result \eqref{eq:skew-asymptotic}, even for $H=0$. Moreover, from Figure \ref{fig:skew-surface-sBergomi}, it seems reasonable to expect that the same asymptotics should hold for the skew at $H=0$. The question remains open, whether it is possible to obtain a precise short-time asymptotic result without using a small vol-of-vol expansion.

As a sanity check, note that when $K$ is the classical Riemann-Liouville kernel we recover the well known constant in the explosion of the skew, see e.g. \cite{BFGHS19,fukasawa2017}.

\section{Numerical analysis}\label{sec:numerics}

\begin{figure}[htbp!]
    \centering
    \begin{subfigure}[t]{0.48\textwidth}
        \centering
        \includegraphics[width=\textwidth]{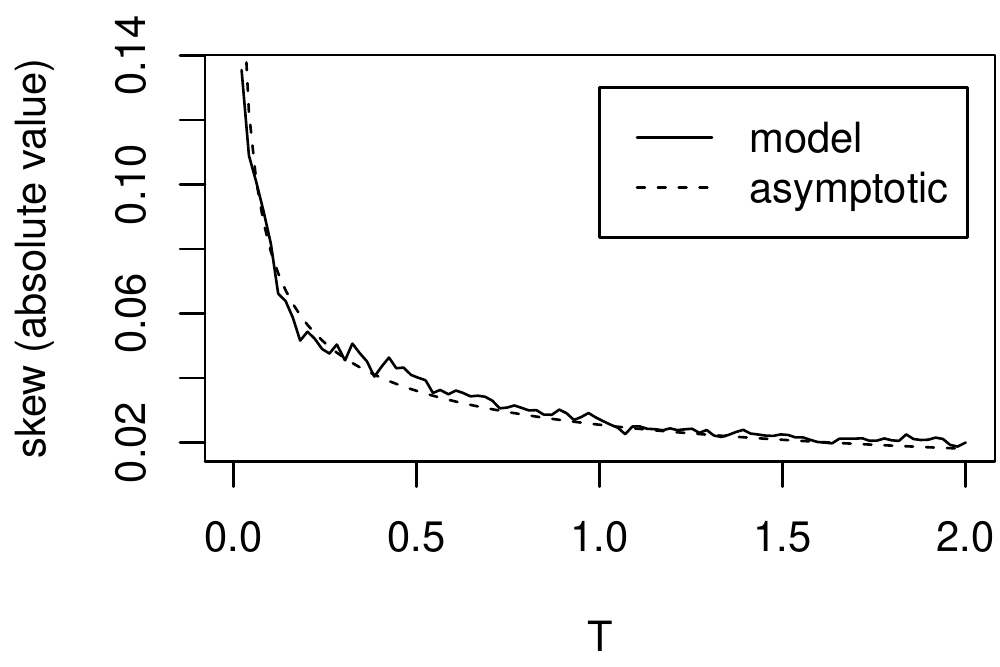}
        \caption{$\eta = 0.2$}
        \label{fig:skews_asymp_small-eta_H0_exact}
    \end{subfigure}%
    \hfill
    \begin{subfigure}[t]{0.48\textwidth}
        \centering
        \includegraphics[width=\textwidth]{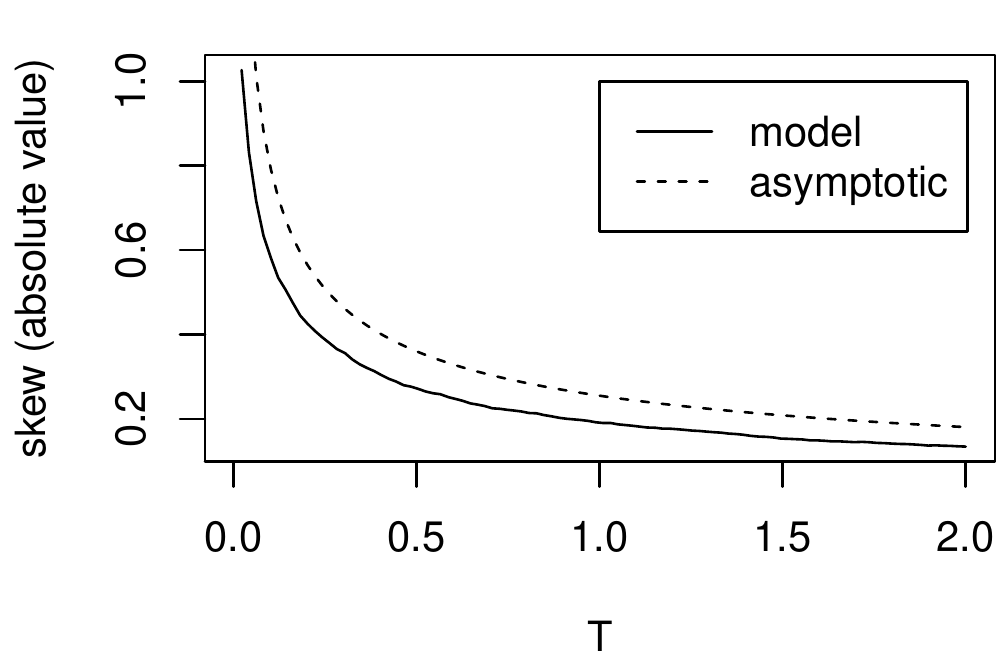}
        \caption{$\eta = 2$}
        \label{fig:skews_asymp_normal-eta_H0_exact}
    \end{subfigure}
    \caption{Asymptotic formula for the ATM-skew for small vol-of-vol in the super-rough Bergomi model with $H = 0$, $\rho = -0.7$, $\zeta = 0.1$, $p = 2$, and small vol-of-vol $\eta = 0.2$ vs ``normal'' vol-of-vol $\eta = 2$. The asymptotic formula is compared against skews computed by Monte Carlo simulation.}
    \label{fig:fig:skew_asymptotic_eta_small_H0}
\end{figure}

We supplement the theoretical results by some numerical experiments. In all these examples, we use the super-rough Bergomi model~\eqref{eq: def of rBergomi dyn}. Skews are computed based on Monte Carlo simulation with exact simulation of the log-fractional Brownian motion~\eqref{eq:log-fbm} together with~\eqref{eq:log-kernel}. More precisely, we compute the covariance function of $(W,\hatW)$ using the formulas in Section~\ref{sec:moments-log-fract} as well as numerical integration for the auto-covariance of $\hatW$. Exact simulation from $(W,\hatW)$ is then done by the Cholesky method. Given samples from the stochastic variance, the asset price process is computed by Euler discretization. We start by comparing the small vol-of-vol expansion with the skews obtained in the model, see Theorem~\ref{thr:rbergomi}. 

In Figure~\ref{fig:fig:skew_asymptotic_eta_small_H0}, we compare the asymptotic formula with the actual skew for two different values of the vol-of-vol parameter $\eta$. Clearly, for small $\eta$ (left), the accuracy is extremely good, and the fit deteriorates noticeably when $\eta$ is increased. Note that we concentrate on the case $H = 0$, as here the behaviour obviously differs most from the rough Bergomi case.
\begin{figure}[H]
    \centering
    \includegraphics[scale=1]{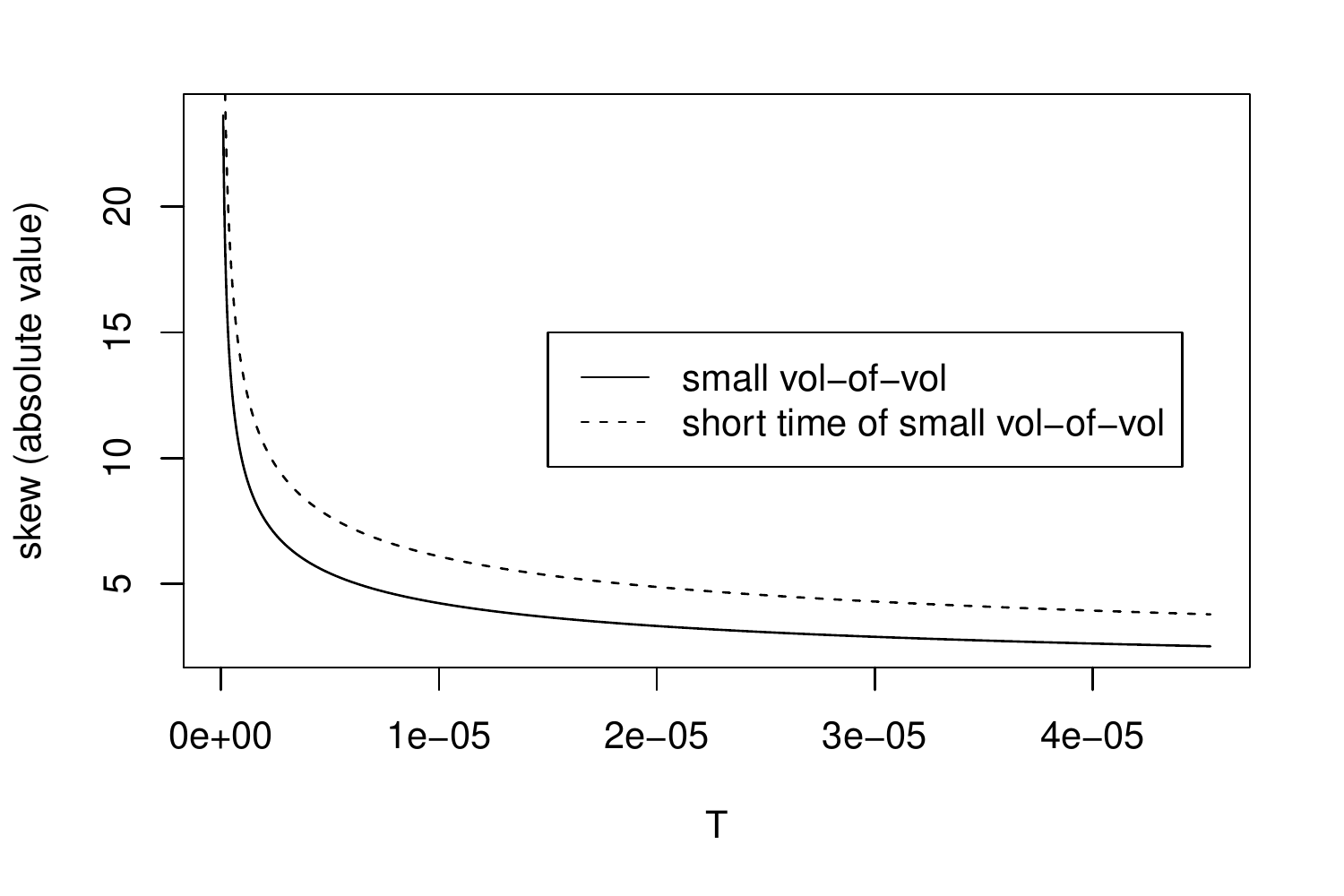}
   \caption{Short time expansion of the asymptotic formula for the ATM- skew for small vol-of-vol, see Theorem~\ref{thr:rbergomi} for the expansion in small vol-of-vol and Corollary~\ref{cor:rbergomi-short-skew} for its short-time expansion}. The parameters correspond to Figure~\ref{fig:skews_asymp_small-eta_H0_exact}.
    \label{fig:skew_asymptotic_eta_small_H0_short}
\end{figure}

Next we consider the short-time asymptotic of the asymptotic skew formula obtained in Theorem~\ref{thr:rbergomi} together with Corollary~\ref{cor:rbergomi-short-skew}, see Figure~\ref{fig:skew_asymptotic_eta_small_H0_short}. We should note that Theorem~\ref{thr:skew} only provides the short time asymptotic for $0<T<\chi = e^{-1/\zeta} \approx 5 \times 10^{-5}$ in our example. Hence, we need to zoom in very closely for the asymptotic formula to hold. The Figure indicates that the convergence of the short-time asymptotics is very slow.

\begin{figure}[H]
    \centering
    \includegraphics[width=\textwidth]{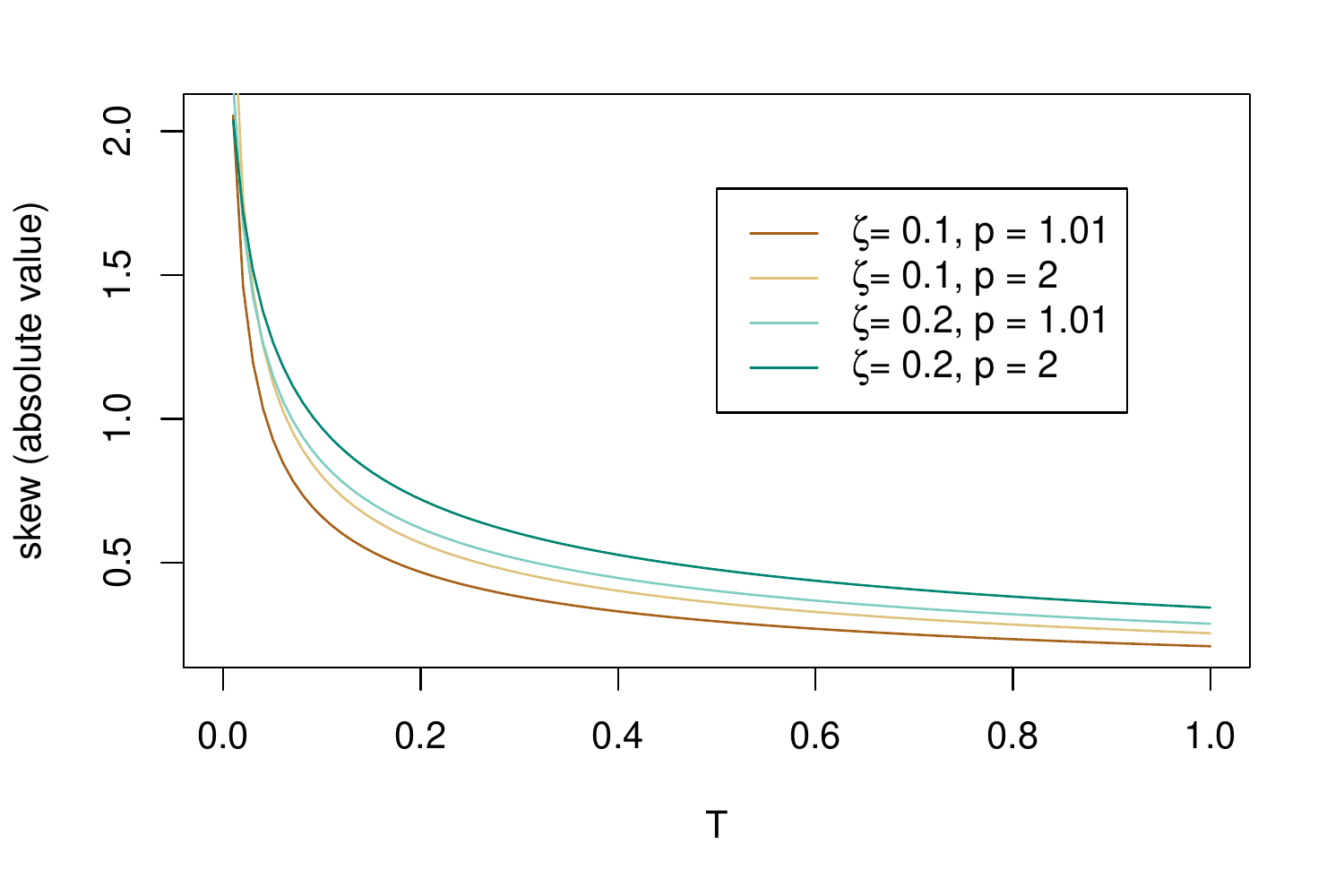}
    \caption{Asymptotic skew formulas for small vol-of-vol in the super-rough Bergomi model for different values of $\zeta$ and $p$. The remaining parameters are $H = 0$, $\rho = -0.7$, $\eta = 2$, $\xi \equiv 0.04$.}
    \label{fig:sberg_p_zeta}
\end{figure}

Coming back to the discussion of the additional parameters $\zeta$ and $p$ in Remarks~\ref{rem:parameter-choice} and~\ref{rem:zeta-p-vol-of-vol}, we compare the small vol-of-vol skew formulas of Theorem~\ref{thr:rbergomi} for different values of $\zeta$ and $p$, see Figure~\ref{fig:sberg_p_zeta}. Clearly, the absolute value of the ATM skew is increasing in both $\zeta$ and $p$, which indicates that one of these parameters could be easily removed -- by fixing it to a canonical value. In this case, we suggest to fix $p$ to a value close to $1$, such as $p = 1.01$ as used in the plot. 

Finally, we compare the super-rough Bergomi model with the standard rough Bergomi model. Figure~\ref{fig:skew_sberg_rberg} compares ATM-skews -- as computed by Monte Carlo simulation -- for both models and different values of $H$. As expected, the curves differ substantially for very small $H$, but move closely together for $H$ large. In this sense, the super-rough Bergomi model can be seen as a perturbation of the rough Bergomi model for $H \gg 0$, which is still well-defined in the limit $H = 0$ -- naturally departing from the rough Bergomi model in the process, i.e., as $H \to 0$.
\begin{figure}[H]
    \centering
    \includegraphics[width=\textwidth]{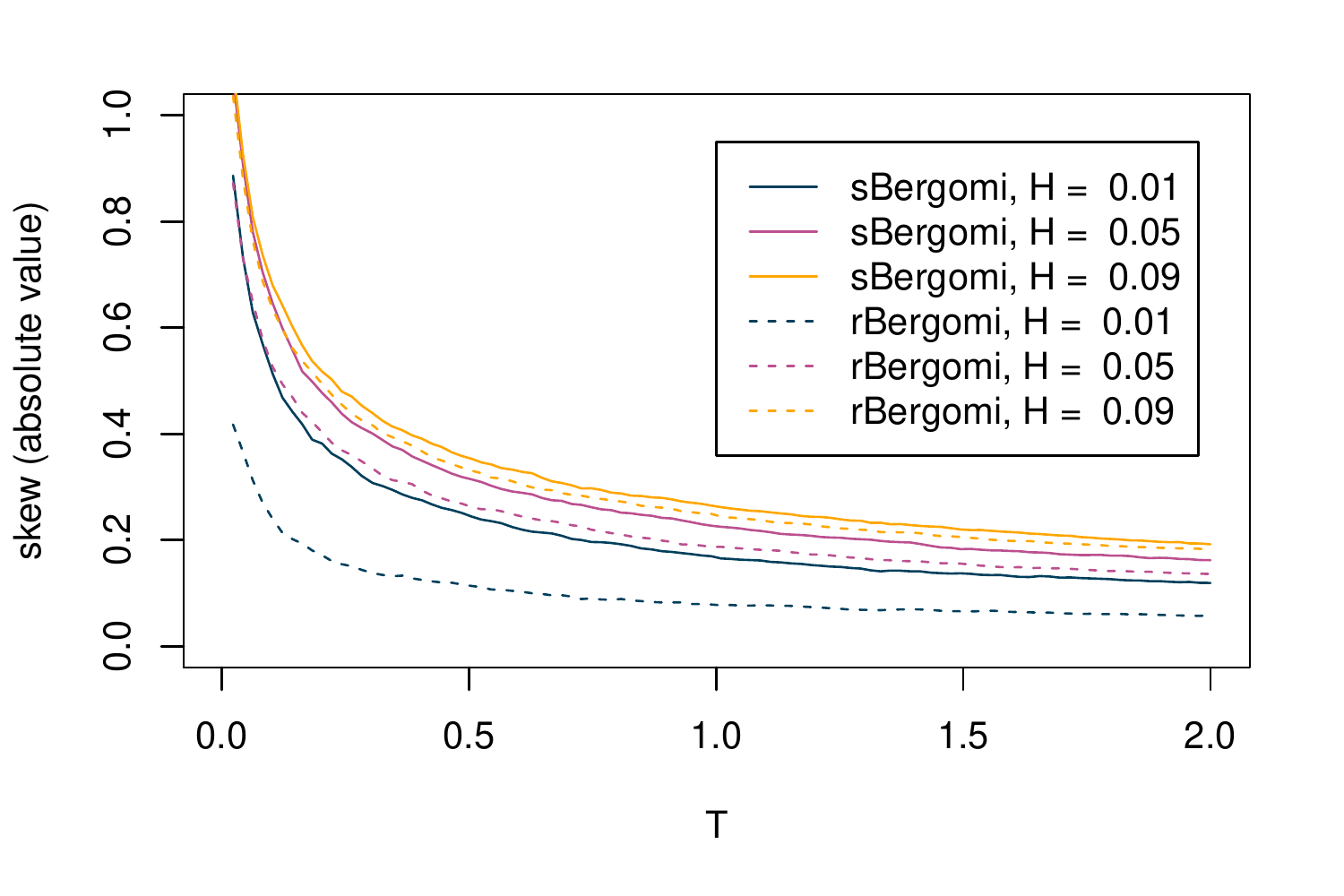}
    \caption{Comparisons between the ATM skews of a rough Bergomi model and a corresponding super-rough Bergomi model for $H \in \{0.01, 0.05, 0.09\}$. Skews are computed by Monte Carlo simulation. The remaining parameters are $\eta = 2$, $\rho = -0.7$, $\xi(t) \equiv 0.04$, and $\zeta = 0.1$, $p = 1.1$.}
    \label{fig:skew_sberg_rberg}
\end{figure}

\bibliographystyle{abbrv}
\bibliography{bib} 

\end{document}